\documentclass[a4paper,12pt]{article}

\usepackage[english]{babel}
\usepackage[utf8]{inputenc}
\usepackage{amsmath}
\usepackage{amsthm}
\usepackage{mathrsfs}
\usepackage{amssymb}
\usepackage{graphicx}
\usepackage[colorinlistoftodos]{todonotes}
\usepackage{xcolor}
\usepackage{verbatim}
\usepackage{theoremref}

\numberwithin{equation}{section}
 
 \newcommand{\be}{\begin{equation}}
 \newcommand{\ee}{\end{equation}}
 
 \newcommand{\taubar}{{\bar \tau}}
 \newcommand{\wtilde}{{\tilde w}}
 \newcommand{\Dtilde}{{\tilde D}}

 \newcommand{\rv}{{\bf r}}
 \newcommand{\pv}{{\bf p}}
 
 \newcommand{\somma}[3]{{\sum_{{#1}={#2}}^{#3}}}

\theoremstyle{plain}
\newtheorem{theorem}{Theorem}[section]
\newtheorem{prop}[theorem]{Proposition}
\newtheorem{defn}[theorem]{Definition}
\newtheorem{cor}[theorem]{Corollary}
\newtheorem{lem}[theorem]{Lemma}

\theoremstyle{remark}
\newtheorem*{oss}{{\bf Remark}}

\newcommand{\uhat}{{\hat{\mathbf u}}}
\newcommand{\utilde}{{\tilde{\mathbf u}}}
\newcommand{\uv}{{\mathbf u}}
\newcommand{\y}{{\mathbf y}}

\newcommand{\nuhat}{{{\hat \nu}_{t,x}}}
\newcommand{\nutilde}{{{\tilde \nu}_{t,x}}}

\newcommand{\E}{{\mathbb E}}
\newcommand{\R}{{\mathbb R}}
\newcommand{\Pbb}{{\mathbb P}}

\newcommand{\Pbbtilde}{{\tilde{\mathbb P}}}

\newcommand{\lamb}{{\lambda_{\beta, \bar p, \tau}}}

 \newcommand{\e}[1]{e^{\scalebox{0.7}{$\displaystyle #1$}}}

\newcommand{\limN}{{\lim_{N \to \infty}}}
\newcommand{\weak}{\rightharpoonup}
\newcommand{\weakstar}{\overset{*}{\rightharpoonup}}

\newcommand{\normp}[2][]{\ensuremath{\left\lVert #1 \right\rVert}_{L^{#2}(Q_T)}}

\newcommand{\expnu}[1][]{\langle #1, \nutilde \rangle}

\begin{document}
\title{Hydrodynamic Limits and Clausius inequality
   for Isothermal Non-linear Elastodynamics with Boundary Tension}
\author{
  Stefano Marchesani\\
Stefano Olla}
\date{\today}
\maketitle

{\bf Keywords:} hydrodynamic limits, isothermal Euler equation,
non-linear wave equation, weak solutions, boundary conditions, entropy solutions,
Clausius inequality, compensated compactness.

{\bf AMS Mathematics Subject Classification:} 82C22, 60K35, 35L50

\begin{abstract}
  We consider a chain of 
  particles connected by 
  an-harmonic springs,
  with a boundary force (tension) acting on the last particle, while the first particle is kept pinned
  at a point. The particles are in contact with stochastic heat baths,
  whose action on the dynamics conserves the volume and the momentum,
  while energy is exchanged with the heat baths in such way that, in equilibrium,
  the system is at a given temperature $T$. We study the space empirical
  profiles of volume stretch and momentum under hyperbolic
  rescaling of space and time as the size of the system growth to
  be infinite, with
  the boundary tension changing slowly in the macroscopic time scale.
  We prove that the probability distributions of these profiles concentrate on $L^2$-valued
  weak solutions of the isothermal Euler equations
  (i.e. the non-linear wave equation, also called p-system), satisfying
  the boundary conditions imposed by the microscopic dynamics.
  Furthermore, the weak solutions obtained satisfy the Clausius inequality between
  the work done by the boundary force and the change of the total free energy in the system.
  This result includes the shock regime of the system. 
\end{abstract}

\section{Introduction}
\label{sec:intro}

Boundary conditions in hyperbolic systems of conservation laws introduce challenging
mathematical problems, in particular for weak solutions that are not of bounded variations.
The solution may depend on the particular approximation used, and reflects different
microscopic origins of the equation. Recently (cf. \cite{marchesani2018note})
we have considered $L^2$-valued weak solution to the
isothermal Euler equation in Lagrangian coordinates on $[0,1]$
(also called in the literature non-linear wave equation or p-system):
\begin{equation} \label{eq:p-intro}
\begin{cases}
 \partial_t  r  -  \partial_x  p =0
 \\
 \partial_t  p - \partial_x \tau( r)=0
\end{cases},
\end{equation}
with the following boundary conditions: $p(t,0)=0$ (the material is attached to a fixed point
on the left side),
$\tau(r(t,1)) = \bar \tau(t)$ (a time dependent force $\bar \tau(t)$ is acting on the right  side).
The precise sense an $L^2$-valued solution satisfies the boundary condition is given in Definition
\ref{def:1}. In \cite{marchesani2018note} we consider viscous approximations
of \eqref{eq:p-intro}, that requires two extra boundary conditions.
We choose these extra boundary conditions to be of Neumann type (i.e. \emph{conservative}).
Adapting the $L^2$-version of the compensated-compactness argument of Shearer \cite{Shearer1}
and \cite{SerreShearer}, we prove in \cite{marchesani2018note}  the existence
of vanishing viscosity solutions to the p-system. Furthermore, these solutions satisfy the usual
Lax-entropy production characterisation and the thermodynamic
Clausius inequality, which relates the change
of the total free energy to the work done by the boundary force
(see Section \ref{sec:clausius-inequality}). We call such weak solutions
\emph{thermodynamic entropy solutions}.

In the present article we study the microscopic \emph{statistical mechanics} origin of \eqref{eq:p-intro}.
We want to understand how equation \eqref{eq:p-intro} emerges in a \emph{hydrodynamic limit},
i.e. a hyperbolic space-time rescaling of a microscopic dynamics.
We consider a chain of $N+1$ particles connected by $N$ anharmonic springs
(see Figure \ref{fig:figura}). The first particle on the left is fixed at a point, while on
the rightmost particle is acting a time-dependent force (tension). 
The Hamiltonian dynamics of this system is perturbed by the action of stochastic
heat baths at temperature $T$. Each heat bath is acting, independently from the others,
between two springs connected by a particle, randomly exchanging momenta and volume stretch.
The energy of the particles is not conserved but exchanged
with the heath baths in such a way that, in equilibrium, the system is at temperature $T$.
The intensity of the action of the heat baths is such that it does not affect the
macroscopic equation directly, but sufficiently strong to provide the required regularity at certain
microscopic scales and establish an \emph{isothermal} macroscopic evolution.
In this sense these heat baths act like a \emph{stochastic viscosity},
vanishing after the space-time hyperbolic rescaling. The conservative nature of these stochastic
heat baths also provides at the boundaries the analogue of the extra Neumann conditions
as used in \cite{marchesani2018note}. 

\begin{figure}
\centering
        \includegraphics[totalheight=4cm]{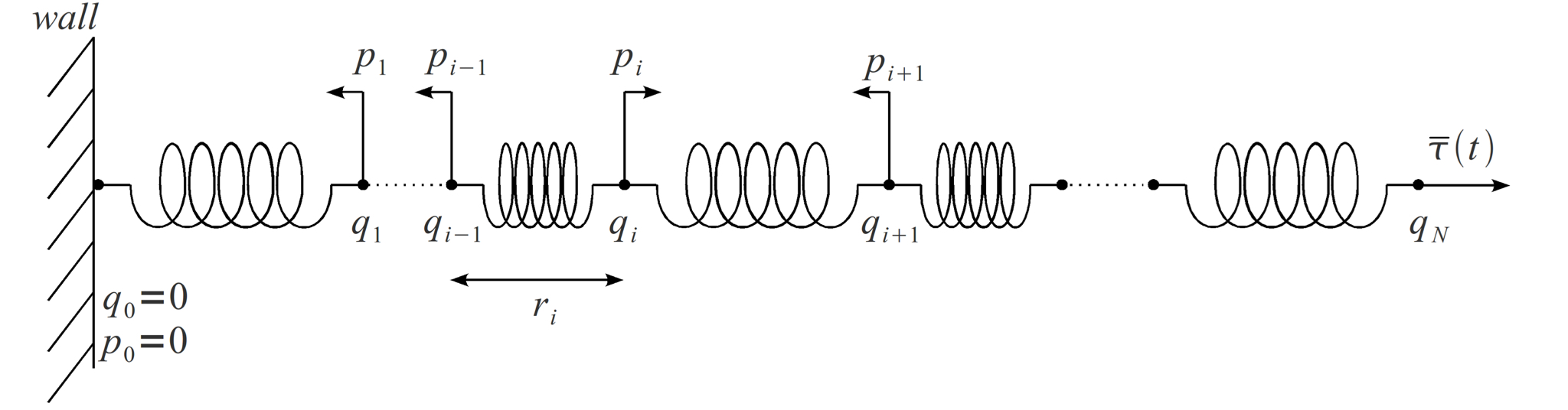}
        \caption{Microscopic model}
    \label{fig:figura}
\end{figure}

We rescale space and time using $N$ as parameter in such a way that the time-dependent external force is changing on the macroscopic scale (i.e. very slowly on the microscopic time scale).
We prove that the probability distributions of the random profiles of volume stretch and momenta
concentrate on the  $L^2$-valued weak solutions of \eqref{eq:p-intro}
(in the sense of Definition \ref{def:1}), that satisfy the Clausius inequality.
Proving uniqueness would complete the convergence theorem.
Unfortunately, uniqueness for such weak solutions is a well known and challenging open problem.

The proof of the convergence to the weak solutions is adapted from the stochastic
version of compensated compactness developed by Fritz in \cite{Fritz1} for the same dynamics but without boundaries (see also Fritz and Toth \cite{toth2004} for a different two component dynamics). 
In a previous work \cite{marchesani2018hydrodynamic},
we considered the same problem as here, but we proved that \eqref{eq:p-intro} were satisfied only in
the bulk by the limit profiles, without giving any information of the boundary conditions,
nor on the \emph{entropic} properties of these solutions.

The main new contributions of the present article are the followings:
\begin{itemize}
\item the limit profiles obtained are $L^2$-valued weak solutions that satisfy the boundary conditions,
  in the sense of Definition \ref{def:1};
  \item the work done by the boundary force is larger than the change in the total free energy (Clausius inequality).  
\end{itemize}

The proof of the Clausius inequality is the content of Section \ref{sec:clausius-inequality}.
It uses the variational characterisation of the (microscopic) relative entropy in order
to connect it to the macroscopic free energy and estimate its time derivative.
In other words, Clausius inequality follows from the microscopic entropy production. 

\section{The Model and the Main Theorem} \label{sec:mainthm}
We study a one-dimensional Hamiltonian system of $N+1$
particles of unitary mass. The position of the $i$-th particle ($i =0,1,\dots, N$)
is denoted by $q_i \in \mathbb R$ and its momentum
by $p_i \in \mathbb R$. We assume that particle $0$ is kept fixed,
i.e. $(q_0,p_0) \equiv (0,0)$, while on particle $N$ is applied a
time-dependent force, $\bar \tau(t)$.

Denote by ${\bf q} =(q_0, \dots, q_N)$ and ${\bf p} =(p_0,\dots, p_N)$.
The interaction between particles $i$ and $i-1$ is described by
the potential energy $V(q_i-q_{i-1})$ of an anharmonic spring, where
$V$ is a  {uniformly convex function that grows
quadratically at infinity}: there exist constants $c_1$ and $c_2$ such
that for any $r \in \mathbb R$: 
\begin{equation}  \label{eq:Vconvex}
0<c_1 \le V''(r) \le c_2.
\end{equation}
Moreover, there are some positive constants $V''_+, V''_-, \alpha$ and $R$ such that
\begin{align} \label{eq:AN}
&\left| V''(r) - V''_+ \right| \le e^{-\alpha r}, \quad r > R
\\
& \left| V''(r) - V''_- \right| \le e^{\alpha r}, \quad r < -R.
\end{align}
For $\tau \in \mathbb R$ and $\beta >0$ we define the canonical Gibbs function as
\begin{equation} \label{eq:Gdefn}
 { G(\tau):= \log \int_{-\infty}^{+\infty} \e{
  -\beta V(r) +\beta \tau r  }\; dr.}
\end{equation}
For $\ell \in \mathbb R$, the free energy is given by the Legendre transform of $G$:
\begin{equation}\label{eq:freedef}
F( \ell) := \sup_{\tau \in \mathbb R} \left\{ \tau \ell - \beta^{-1}G( \tau)\right\},
\end{equation}
so that its inverse is
\begin{equation}\label{eq:legInv}
G(\tau) = \beta \sup_{\ell \in \mathbb R} \left\{ \tau \ell -  F(\ell)\right\}.
\end{equation}
Note that we neglect to write the dependence of $F$ and $G$ on $\beta$, as it shall be fixed throughout the paper.

We denote by $\ell(\tau)$ and $\tau(\ell)$  the corresponding convex conjugate variables, that depend parametrically on $\beta$ and satisfy
\begin{equation}\label{eq:gamrho}
\ell(\tau) = \beta^{-1} G'(\tau), \quad  \tau(\ell)  =  F'(\ell).
\end{equation}
We identify $\tau(\ell)$  with the equilibrium tension of the system of length $\ell$,
and we assume that the potential $V$, besides satisfying the assumptions above,
is such that $\tau$ is strictly convex (i.e. $\tau''(\ell) >0$ for all $\ell \in \mathbb R$).
\begin{oss}
  {At the present time, we do not know a general condition on $V$
    that yield a strictly convex tension, but as example {(cf \cite{marchesani2018hydrodynamic}, Proposition A.7)} one 
    may take $V$ to be a mollification of the function} 
\begin{equation} \label{eq:Vvera}
r \longmapsto \frac{1}{2}(1-\kappa)r^2+\frac{1}{2}\kappa r|r|_+,
\end{equation}
 {where $|r|_+=\max\{r,0\}$ and $\kappa \in (0,1/3)$.} 
\end{oss}
Define the Hamiltonian:
\begin{equation} 
\label{eq:hamiltonian}
\mathcal H_N({\bf q}, {\bf p},t):= \sum_{i =0}^N \left( \frac{p_i^2}{2}+V(q_i-q_{i-1})\right) - q_N \taubar(t),
\end{equation}
where $\taubar(t)$ is the external tension. Since the interaction depends only on the distance between particles, we define
\begin{equation}
r_i := q_i-q_{i-1}, \quad i = 1,\dots, N.
\end{equation}
Consequently,  recalling that $p_0 \equiv 0$, the configuration of the system is given by  
$({\bf  r}, {\bf p} ):=(r_1, \dots, r_N, p_1,\dots,p_N)$ and the
  phase space is  $\mathbb R^{2N}$. Thus, the Hamiltonian reads
\begin{align}
\mathcal H_N({\bf r}, {\bf p},t) = \sum_{i=1}^N \left( \frac{p_i^2}{2}+V(r_i)- \taubar(t)r_i \right).
\end{align}
 {We add to the Hamiltonian dynamics physical and artificial noise. Thus, the full dynamics of the system is determined by the generator}
\begin{equation} \label{eq:fullgen}
 \mathcal G_N^{\bar \tau(t)}:= N L_N^{\bar \tau(t)}+ N{\sigma} 	\left( S_N + \tilde  S_N^{\taubar(t)} \right).
\end{equation}
 {$\sigma=\sigma(N)$ is a positive number that tunes the strength of the noise. We take it such that}
\begin{equation} \label{eq:siglim}
\lim_{N\to \infty} \frac{{\sigma}}{N}= \lim_{N \to \infty} \frac{N}{{\sigma}^2} = 0.
\end{equation}
The Liouville operator $L_N^{\bar \tau(t)}$ is given by
\begin{equation} \label{eq:hamgen}
L_N^{\bar \tau(t)}  = \sum_{i = 1}^N (p_i-p_{i-1}) \frac \partial {\partial r_i} +
\sum_{i =1}^{N-1} \left(V'(r_{i+1}) -V'(r_i) \right) \frac \partial {\partial p_i} +(\bar \tau(t) - V'(r_N)) \frac \partial {\partial p_N},
\end{equation}
together with $p_0 \equiv 0$. Note that the time
scale in the tension is chosen such that it changes smoothly on the
macroscopic scale. 

The operators $S_N$ and $\tilde S_N^{\taubar(t)}$ generate the stochastic part of the dynamics,
{modelling the interaction with a heat bath at constant temperature $\beta^{-1}$,} and are defined by
\begin{equation} \label{eq:SN}
S_N := -\sum_{i=0}^{N-1}  D^*_i D_i, \quad \tilde S_N^{\taubar(t)} := -\sum_{i=1}^N  \tilde D^*_i \tilde D_i,
\end{equation}
where, for $i=1,\dots, N-1$,
\begin{align} \label{eq:DD*}
& D_i := \frac{\partial }{\partial p_{i+1}}-\frac{\partial }{\partial p_i}, \qquad D^*_i := p_{i+1}-p_i - \beta^{-1}D_i
\\
&\tilde D_i  := \frac{\partial }{\partial r_{i+1}}-\frac{\partial }{\partial r_i}, \qquad \tilde D^*_i := V'(r_{i+1})- V'(r_i) -  \beta^{-1}\tilde D_i.
\end{align}
The extra boundary operators were first considered in \cite{MarchesaniDiffusive} and are given by
\begin{align}
D_0 &:= \frac{\partial}{\partial p_1}, \qquad D_0^* := p_1 - \beta^{-1} D_0,
\\
\Dtilde_N &:= -\frac{\partial}{\partial r_N}, \qquad \Dtilde_N^* := \taubar(t)-V'(r_N)-\beta^{-1}\Dtilde_N.
\end{align}

On the one-particle state space $\mathbb R^2$ we define a family of probability measures
\begin{equation} \label{eq:lambdagamdef}
\lamb (d r, dp) :=\frac{1}{\sqrt{2\pi\beta^{-1}}}  \e{-\frac{\beta}{2}(p-\bar p)^2-\beta V(r) +\beta\tau r - G(\beta,\tau) }\; dr \, dp.
\end{equation}
The mean elongation and momentum are
\begin{equation}\label{eq:rhopi}
{\int r \ d\lamb  
  = \ell(\tau), \quad \int p \ d\lamb = \bar p.}
\end{equation}
We also have the relations
\begin{equation}\label{eq:temper}
{ \int V'(r) \ d\lamb = \tau, \quad \int p^2 \ d\lamb -\bar p^2 = \beta^{-1},}
\end{equation}
that identify  $\tau$ as the tension and $\beta^{-1}$ as the temperature.

Define the family of product measures $ \lambda^N_t = \lambda^N_{\beta,0, \bar \tau(t)}$ where
\begin{equation}\label{eq:invmeas}
  \lambda^N_{\beta,\bar p,  \tau}(d{\bf r}, d{\bf p}) = \prod_{i=1}^N\lambda_{\beta,\bar p,  \tau}(dr_i, dp_i).
\end{equation}
Notice that $L_N^{\bar \tau(t)}$ is antisymmetric with respect to $ \lambda^N_t $, while
$S_N$ and $\tilde S_N^{\taubar(t)}$ are symmetric. It follows that, in the case  $\bar \tau$ is constant
in time, $\lambda^N_{\beta,0, \bar \tau}$ is the unique stationary measure for the dynamics.
This is the canonical Gibbs measure
at a temperature $\beta^{-1}$, pressure $ \bar \tau$ and velocity $0$. 


Define the discrete gradients and Laplacian by
\begin{align}
\nabla a_i := a_{i+1}-a_i, \qquad \nabla^* a_i := a_{i-1}-a_i,
\\
\Delta a_i := -\nabla \nabla^* a_i = -\nabla^* \nabla a_i= a_{i+1}+a_{i-1}-2a_i.
\end{align}
The time evolution of the system is described by the following
system of stochastic differential equations
\begin{equation} \label{eq:SDE}
\begin{cases}
d r_1 =  Np_1 d t + \sigma N \nabla V'(r_1) d t - \sqrt{2 \beta^{-1}\sigma N}  \,  d \widetilde w_1
\\
d r_i =  -N\nabla^* p_i- d t + \sigma N \Delta V'(r_i) d t + \sqrt{2 \beta^{-1} \sigma N} \,
\nabla^* d\widetilde w_i, & 2 \le i \le N-1\\
d r_N =  -N\nabla^* p_N dt + \sigma N \left(\taubar(t) + V'(r_{N-1}) - 2V'(r_N)\right)
+ \sqrt{2 \beta^{-1} \sigma N} \, \nabla^*  d\widetilde w_N,
\\
d p_1 = N \nabla V'(r_1) d t + \sigma  N\left(p_2 - 2p_1\right) dt + \sqrt{2 \beta^{-1} \sigma N}
\,  \nabla^* dw_{1},\\
d p_j =  N\nabla V'(r_j) d t +\sigma N \Delta p_j  d t + \sqrt{2 \beta^{-1} \sigma N}
\,  \nabla^* dw_j, & 2 \le j \le N-1
\\
d p_N =N (\taubar(t)-V'(r_N)) d t -\sigma  N\nabla p_{N-1}d t
+ \sqrt{2\beta^{-1}\sigma N }\, d  w_{N-1}
\end{cases}
\end{equation}
Here  $\{w_i\}_{i=0}^\infty$, $\{\widetilde w_i\}_{i=1}^\infty$ are
independent families of independent Brownian motions on a
common probability space $(\Omega, \mathbb P)$.
The expectation with respect to $\Pbb$ is denoted by $\E$.

Notice that the noise introduced by the heat bath respects the \emph{boundary conditions}
$V'(r_{N+1}(t)) = \bar\tau(t)$ and $p_0(t) = 0$ already present in the Hamiltonian part
of the dynamics, while it introduces the \emph{Neumann type boundary conditions}
$r_0(t) = r_1(t)$ and $p_{N+1}(t) = p_N(t)$. In this sense these boundary conditions
are the microscopic analogous of those taken in the viscous approximation used in reference
\cite{marchesani2018note}.

Thanks to the assumptions we made on the interaction $V$, it is possible to show (cf \cite{Fritz1} or Appendix A of \cite{marchesani2018hydrodynamic}) the following
\begin{prop}
For any fixed $\beta>0$, the application $\tau : \mathbb R \to \mathbb R$ is smooth and has the following properties:
\begin{itemize}
\item[i)] $c_2^{-1} \le \tau(\ell) \le c_1^{-1}$ for all $\ell \in \mathbb R$;
\item[ii)] $\tau'', \tau''' \in L^2(\mathbb R) \cap L^\infty(\mathbb R)$.
\end{itemize}
\end{prop}
Furthermore, we assume $\tau''(\ell) >0$ for all $\ell \in \R$.
\begin{oss}
The condition $\tau' \ge c_2^{-1}>0$ is a condition of strict hyperbolicity. On the other hand,  $\tau'' >0$ is a condition of genuine nonlinearity, and it is easy to see that it rules out symmetric interactions ($V(-r)=V(r)$). Nevertheless, such a condition may be relaxed as in \cite{SerreShearer} and we can allow $\tau''$ to vanish at most at one point, which is compatible with having a symmetric interaction.
\end{oss}


Denote by $\mu_t^N$ the probability measure of the system a time $t$. 
Then, the density $f_t^N$ of $\mu_t^N$ with respect to  $\lambda^N_t $ solves the Fokker-Plank equation
\begin{equation}\label{eq:ftN}
\frac{\partial}{\partial t} \left( f_t^N \lambda_t^N \right)= \left(\mathcal G_N^{\bar \tau(t),*} f_t^N \right) \lambda^N_t.
\end{equation}
Here 
\begin{equation}
\mathcal G_N^{\bar \tau(t),*}  = -N L_N^{\bar \tau(t)}  +N \sigma \left(S_N+ S_N^{\taubar(t)}\right)
\end{equation}
is the adjoint of $\mathcal G_N^{\bar \tau(t)}$ with respect to $\lambda^N_t$.

{
We define the relative entropy
\begin{equation} \label{eq:HN}
H_N(t) :=\int f_t^N \log f_t^N d \lambda^N_t,
\end{equation}
and require that the initial distribution $f_0^N$ is such that $H_N(0) \le CN$ for some $C$ independent of $N$.
}

We are interested in the
macroscopic behaviour of the volume stretch and momentum of the particles,
at time $t$, as $N \to \infty$. Note that $t$ is already the
macroscopic time, as we have already multiplied by $N$ in the
generator. We shall use Lagrangian coordinates, that is our space
variables will belong to the lattice $\{1/N, \dots, (N-1)/N, 1\}$. 

 {Consequently, we set ${\bf u}_i := (r_i, p_i)$. For a fixed macroscopic time $T$, we  introduce the empirical measures on $[0,T] \times [0,1]$ representing the space-time distributions on the interval $[0,1]$ of volume stretch and momentum: }
\begin{equation}\label{eq:empmeas}
{\pmb \zeta}_N(dx,dt) := \frac{1}{N} \sum_{i=1}^N \delta \left(x-\frac{i}{N}\right)  \uv_i(t) dx\; dt.
\end{equation}
We expect that the measures ${\pmb \zeta}_N(dx,dt)$ converge, as $N \to
\infty$ to an absolutely continuous measure with density $(r(t,x),p(t,x))$, satisfying the
following system of conservation laws: 
\begin{equation} \label{eq:psystem}
\begin{cases}
 \partial_t  r (t,x) -  \partial_x  p(t,x)=0
 \\
 \partial_t  p(t,x)- \partial_x \tau( r(t,x))=0
\end{cases},
\qquad \begin{matrix} p(t,0)=0, & \tau(r(t,1)) = \bar \tau(t) \\ r(0,x) = r_0(x) & p(0,x)=p_0(x)\end{matrix}.
\end{equation}
Since \eqref{eq:psystem} is a hyperbolic system of nonlinear partial
differential equations, its solutions may develop shocks in a finite
time, even if smooth initial data are given. Therefore, we shall
look for \emph{weak} solutions, which are defined even if
discontinuities appear.

\begin{defn} \thlabel{def:1}
Fix $T >0$ and let $Q_T = [0,T]\times [0,1]$. We say that $(r, p) \in L^2(Q_T)$
is a \emph{weak solution} of 
\eqref{eq:psystem} provided
\begin{equation} \label{eq:maineqr1}
\int_0^1\varphi(0,x)r_0(x)dx+\int_0^T \int_0^1 \left(  r\partial_t \varphi  - p \partial_x\varphi  \right) dx dt=0
\end{equation}
\begin{equation} \label{eq:maineqp1}
\int_0^1\psi(0,x) p_0(x)dx+  \int_0^T \int_0^1    \left( p\partial_t\psi - \tau(r) \partial_x\psi \right) dxdt
  + \int_0^T  \psi(t,1) \taubar(t) dt=0
\end{equation}
for all  functions $\varphi, \psi \in C^2(Q_T)$ such that $\varphi(t,1)  = \psi(t,0) = 0$ for all $ t \in [0,T]$.
\end{defn}

Denote by $\frak Q_N$ the probability distribution of ${\pmb \zeta}_N$ on $\mathcal M(Q_T)^2$. Observe that  ${\pmb \zeta}_N\in  C([0,T], \mathcal M([0,1])^2)$, where $\mathcal M([0,1])$ 
is the space of signed measures on $[0,1]$, endowed by the weak topology.
Our aim is to show the convergence
\begin{equation}\label{eq:conv}
{\pmb \zeta}_N(J) \to \left( \int_0^T \int_0^1 J(t,x)r(t,x) dx dt, \int_0^T\int_0^1 J(t,x) p(t,x) dx dt \right),
\end{equation}
where $r(t,x)$ and $p(t,x)$ satisfy \eqref{eq:maineqr1}-\eqref{eq:maineqp1}. 
 Since we do not have uniqueness for the solution of these equations, we need 
a more precise statement. 
\begin{theorem}[Main theorem] \thlabel{thm:main}
Assume that the initial distribution satisfies the entropy condition $H_N(0) \le C N$. Then sequence $\mathfrak Q_N$ is compact and any limit point of $\mathfrak Q_N$ 
has support on absolutely continuous measures with densities $r(t,x)$ and $p(t,x)$
solutions of \eqref{eq:maineqr1}-\eqref{eq:maineqp1} and belonging to $L^\infty(0,T;L^2(0,1))$.

 Moreover, if the system at time $t=0$ is at a local equilibrium, namely if
\begin{align}
f_0^N({\bf r}, {\bf p}) :=\prod_{i=1}^N\e{\beta\left[\left(\tau\left( r_0 \left(\frac i N \right)\right)-\bar \tau(0) \right)r_i+ p_0\left( \frac i N \right)p_i \right]-G \left( \tau \left(r_0 \left( \frac i N\right) \right) \right)+G(\bar \tau(0)) -\frac \beta 2 p_0\left(\frac i N \right)^2},
\end{align}
then we have the following Clausius inequality
\begin{equation}
  \begin{split}
   \mathbb E^{\mathfrak Q}\left(\int_0^T  \int_0^1 [\mathcal F(t,y)  -  \mathcal F(0,y)] dy  \right)
    \le  \mathbb E^{\mathfrak Q}\left( \int_0^TW(t)dt \right),
  \end{split}
\end{equation}
where
\begin{align}
\mathcal F(t,x) := \frac{p(t,x)^2}{2}+F(r(t,x))
\end{align}
is the free energy and
\begin{equation}
 W(t) := - \int_0^t \bar \tau'(s) \int_0^1 r(s,y)dy+ \bar \tau(t) \int_0^1 r(t,y)dy- \bar \tau(0) \int_0^1 r_0(y)dy
\end{equation}
is the work done by the tension $\bar\tau$.
\end{theorem}

Notice that in the case the total length $L(t) = \int_0^1 r(t,y)dy$ is time differentiable, the definition of
work coincide with the usual one: $W(t) := \int_0^t \bar \tau(s) L'(s) ds$.

\section{Some bounds from Relative entropy and Dirichlet forms}
\label{sec-bounds}

Define the Dirichlet forms 
\begin{align} \label{eq:DN}
\begin{split}
& \mathcal D_N(t) := \int {f_t^N} (-S_N \log {f_t^N}) d\lambda^N_t =  \int \frac{1}{f_t^N}  \left[
\left(\frac{\partial f_t^N}{\partial p_1}\right)^2 +  \sum_{i=1}^{N-1}
\left(\frac{\partial f_t^N}{\partial p_{i+1}}- \frac{\partial f_t^N}{\partial p_i}\right)^2\right] d\lambda^N_t,
\\
& \widetilde {\mathcal D}_N(t) := \int {f_t^N} (-\tilde S_N^{\bar\tau(t)} \log {f_t^N}) d\lambda^N_t =
\int \frac{1}{f_t^N} \left[ \sum_{i=1}^{N-1}
\left(\frac{\partial f_t^N}{\partial r_{i+1}}- \frac{\partial f_t^N}{\partial r_i}\right)^2 +   
\left(\frac{\partial f_t^N}{\partial r_N}\right)^2 \right] d\lambda^N_t.
\end{split}
\end{align}
\begin{prop} \thlabel{prop:relative}
The following inequality holds for any $t \ge 0$:
\begin{align} \label{eq:entr0}
H_N(t)-H_N(0)\le 
                        - \beta \int_0^t \taubar'(s) \int \sum_{i=1}^N r_i(s)  f_s^N  d \lambda^N_s+N\beta \int_0^t \bar \tau'(s)\ell(\bar\tau(s))ds.
\end{align}

Moreover, there is  $C(t)$ independent of $N$ such that
\begin{align} \label{eq:entr00}
H_N(t) +N \sigma \beta^{-1} \int_0^t\left( \mathcal D_N(s)+\tilde{\mathcal D}_N(s) \right)ds \le C(t) N.
\end{align}
\end{prop}
\begin{proof}
\begin{align}
  \frac{d H_N(t)}{dt} =& \int \partial_t (f_t^N  \lambda^N_t)\log f_t^N d \rv d \pv
                        + \int \lambda_t^N \partial_t f_t^N d \rv d \pv
\\
                    =& \int \left( \mathcal G_N^{\taubar(t),*} f_t^N \right) \log f_t^N d \lambda^N_t
                      - \int f_t^N \partial_t \lambda^N_t d \rv d \pv + \partial_t \int f_t d \lambda^N_t 
\\
                       =& \int f_t^N \mathcal G_N^{\taubar(t)} \log f_t^N d \lambda^N_t-
                        \beta \taubar'(t) \int \left[q_N - N\ell(\bar\tau(t))\right] f_t^N  d \lambda^N_t 
\\
                       =& -  N \sigma \beta^{-1} \left( \mathcal D_N(t) + \tilde{\mathcal D}_N(t) \right)
                        - \beta \taubar'(t) \int \left[q_N - N\ell(\bar\tau(t))\right] f_t^N  d \lambda^N_t ,
\end{align}
since $\int f_t^N L_N^{\taubar(t)} \log f_t^N d \lambda^N_t =0$, by the antisymmetry of $L_N^{\taubar(t)}$.
Thus, \eqref{eq:entr0} follows after an integration in time and recalling that $q_N = \sum_{i=1}^N r_i$ and that the Dirichlet forms are non-negative.

By the entropy inequality and the strict convexity of $G(\tau)$ we have, for any $\alpha >0$,
\begin{align}
  - \beta \taubar'(t) \int &\left[q_N - N\ell(\bar\tau(t))\right] f_t^N  d \lambda^N_t \le
  \alpha^{-1} H_N(t) + \alpha^{-1} \log \int e^{ - \alpha\beta \taubar'(t) \left[q_N - N\ell(\bar\tau(t))\right]} d \lambda_t^N \nonumber
\\
  &= \alpha^{-1} H_N(t) + \alpha^{-1}N \left[ G(\bar\tau(t) -\alpha\bar\tau'(t)) - G(\bar\tau(t)) + \alpha\bar\tau'(t) \ell(\bar\tau'(t))\right] \nonumber
  \\
  &\le  \alpha^{-1} H_N(t) + \alpha\bar\tau'(t)^2 N C_G.  
\end{align}
By choosing $\alpha = |\bar\tau'(t)|^{-1}$ we obtain the bound
\begin{equation}
  \label{eq:2}
  \begin{split}
    \frac{d}{dt} H_N(t) + N \sigma \beta^{-1}\left( \mathcal D_N(t) + \tilde{\mathcal D}_N(t) \right) 
    \le |\bar\tau'(t)| \left(H_N(t) + NC_G\right).
\end{split}
\end{equation}
By Gronwall's inequality we get
\begin{align}
  \label{eq:3}
  H_N(t) + N \sigma \beta^{-1} \int_0^t \left( \mathcal D_N(s) + \tilde{\mathcal D}_N(s) \right) ds
&  \le H_N(0) e^{\int_0^t |\tau'(s)| ds} + NC_G \int_0^t  |\tau'(s)| e^{\int_s^t |\tau'(u)| du} ds\nonumber
   \\
 & \le C(t) N.
\end{align}

\end{proof}

Observe that $C(t)$ in this proposition is equal to $C_0$ if $\bar\tau'(t) = 0$,
and that  can be chosen independent of $t$ if
$\bar\tau'(t) = 0$ for $t> t_0$ for some $t_0$. 

The energy bound is a standard consequence of the bound on the relative entropy.
\begin{prop}[Energy estimate] \thlabel{prop:energy}
  For any $t \ge 0$
\begin{align}
 \int \left[ \sum_{i=1}^N \left( \frac{p_i^2}{2}+V(r_i)\right)\right] f^N_t d\lambda^N_t \le C(t) N,
\end{align}
and the constant $C(t)$ can be chosen independent of $t$ in the case $\bar\tau'(t) = 0$ for $t\ge t_0$.
\end{prop}
\begin{proof}
By the entropy inequality and for  $0<\alpha  < \beta$,
\begin{align}
  \int \left[ \sum_{i=1}^N \left( \frac{p_i^2}{2}+V(r_i)\right)\right] f_t d\lambda_t  
 \le \frac{1}{\alpha} H_N(t)+\frac{1}{\alpha}\log \int \e{\alpha\sum_{i=1}^N \left(\frac{p_i^2}{2}+V(r_i) \right)} d \lambda^N_t
 \\ \nonumber
  =  \frac{1}{\alpha} H_N(t)+\frac{N}{\alpha}\log \int \e{(\alpha-\beta)\left( \frac{p_1^2}{2}+ V(r_1) \right)
   +\beta\taubar(t)r_1- G(\taubar(t)) }dr_1 \frac{d p_1}{\sqrt{2\pi \beta^{-1}}}.
\end{align}
Note that thanks to our choice of $\alpha$ the last integral is convergent and bounded with respect to $t$.
Thus, the conclusion follows as a consequence of \thref{prop:relative}.
\end{proof}

\section{The hydrodynamic limit}

\subsection{Microscopic solutions}
Let $(r_i(t), p_i(t))_{i=1}^N$ be solutions of  \eqref{eq:SDE}.
Let $\varphi,\psi  \in C^2(Q_T)$ be such that
$\varphi(t,1)=\psi(t,0)=0$ for all $t \in [0,T]$, and let
\begin{equation}
\varphi_i(t):= \varphi \left( t, \frac{i}{N} \right), \quad \psi_i(t):= \psi \left( t, \frac{i}{N} \right), \qquad i=1,\dots, N.
\end{equation}
We set $V'_i:=V'(r_i)$ and evaluate
\begin{align}
 &\frac{1}{N}\somma{i}{1}{N} \varphi_i(T)r_i(T)-\frac{1}{N} \somma{i}{1}{N}\varphi_i(0)r_i(0) = \int_0^T\frac{1}{N}\somma{i}{1}{N}\dot \varphi_i(t) r_i(t) dt+
\\
&+  \int_0^T  \left(p_1 \varphi_1- \somma{i}{2}{N-1} \nabla^* p_i\varphi_i +\nabla^* p_N\varphi_N\right) dt+ 
\nonumber\\
&+ \sigma \int_0^T \left(\varphi_1\nabla V'_1+\somma{i}{2}{N-1}\varphi_i \Delta V'_i+\varphi_N(\taubar(t)+V'_{N-1}-2V'_N) \right) dt+
\nonumber \\
&+ \sqrt{2\beta^{-1}  \frac{\sigma}{N}} \int_0^T \left( -\varphi_1 d \wtilde_1+ \somma{i}{2}{N-1}\varphi_i \nabla^* d \wtilde_i+ \varphi_N \nabla^* d\wtilde_N\right). \nonumber
\end{align}
We use the summation by parts formula
\begin{align}
\somma{i}{2}{N-1} \varphi_i \nabla^* a_i & =  \somma{i}{1}{N-1}a_i\nabla \varphi_i+\varphi_1 a_1
- \varphi_N a_{N-1}
\end{align}
with $a_i = p_i$, $a_i = V'_{i+1}- V'_i$ and $a_i = d \wtilde_i$ in order to obtain
\begin{align}
& \frac{1}{N}\somma{i}{1}{N} \varphi_i(T)r_i(T)-\frac{1}{N} \somma{i}{1}{N}\varphi_i(0)r_i(0) \nonumber
\\
 =& \int_0^T\frac{1}{N}\somma{i}{1}{N}\dot \varphi_i(t) r_i(t) dt - \int_0^T \somma{i}{1}{N-1} \nabla p_i \varphi_i  dt+ \int_0^T \varphi_N p_N dt+
\\
&  - \sigma \int_0^T \somma{i}{1}{N-1} \nabla \varphi_i\nabla V'_i dt +\sigma \int_0^T \varphi_N(\taubar(t)-V'_N)dt+
\nonumber \\
& +\sqrt{2 \beta^{-1}\frac{\sigma}{N}} \int_0^T \somma{i}{1}{N-1}\nabla \varphi_i d \wtilde_i
\nonumber \\
=& \int_0^T\frac{1}{N}\somma{i}{1}{N}\dot \varphi_i(t) r_i(t) dt- \int_0^T \somma{i}{1}{N-1} \nabla \varphi_i p_i  dt - \sigma \int_0^T \somma{i}{1}{N-1}\nabla \varphi_i \nabla V'_i dt + \label{eq:prova}
\\
& +\sqrt{2 \beta^{-1} \frac{\sigma}{N}} \int_0^T \somma{i}{1}{N-1}\nabla \varphi_id \wtilde_i,
\nonumber
\end{align}
as $\varphi_N(t) = \varphi(t,1)=0$. After a second summation by parts we obtain
\begin{align}
-  \sum_{i=1}^{N-1}\nabla \varphi_i \nabla V'_i = \sum_{i=1}^{N-1}V'_i\Delta \varphi_i  + V'_1\nabla \varphi_0-V'_N \nabla \varphi_{N-1}
\end{align}
so that we can write
\begin{align}
 \frac{1}{N}\somma{i}{1}{N} \varphi_i(T)r_i(T)-\frac{1}{N} \somma{i}{1}{N}\varphi_i(0)r_i(0)  =  \int_0^T\frac{1}{N}\somma{i}{1}{N}\dot \varphi_i(t) r_i(t) dt - \int_0^T \somma{i}{1}{N-1} p_i \nabla \varphi_i   dt + \tilde R_N,
\end{align}
with
\begin{align} \label{eq:Rtilde}
\tilde R_N =& \sigma \int_0^T  \sum_{i=1}^{N-1}V'_i(t) \Delta \varphi_i dt+\sqrt{2 \beta^{-1} \frac{\sigma}{N}} \int_0^T \somma{i}{1}{N-1}\nabla \varphi_i(t) d \wtilde_i(t)
\\
&+ \sigma\int_0^T V'_1(t) \nabla \varphi_0 dt -\sigma\int_0^T V'_N(t) \nabla \varphi_{N-1} dt.  \nonumber
\end{align}
\begin{lem} \thlabel{lem:rtilde}
$\mathbb E[|\tilde R_N|^2] \to 0$.
\end{lem}
\begin{proof}
  Since $\varphi \in C^2([0,1])$,  we can estimate
  $|\Delta \varphi_i| \le \| \partial^2_{xx}\varphi\|_\infty N^{-2}$.
  Moreover, since $c_1 \le V''(r) \le c_2$ we have $|V'(r)|^2 \le C \left(1+r^2 \right)$, and
\begin{align}
  \mathbb E \left[\left | \int_0^T \sigma \sum_{i=1}^{N-1} V'_i \Delta \varphi_i dt \right |^2 \right]\le
  \sigma^2T \left(\sum_{i=1}^{N-1} (\Delta \varphi_i)^2\right)
  \int_0^T \mathbb E \left[ \sum_{i=1}^{N-1} (V'_i(t))^2 \right] dt \\ \nonumber
  \le  C\frac{\sigma^2T}{N^3}
  \int_0^T \mathbb E \left[ \sum_{i=1}^{N-1}\left(1+r_i^2(t) \right) \right] dt \le C(T) \frac{\sigma^2}{N^2}.
\end{align}
Since the Brownian motions $d \wtilde_i$ are independent, we evaluate
\begin{align}
  \mathbb E \left[ \left |\sqrt{\frac{2\sigma}{\beta N}}
  \int_0^T \somma{i}{1}{N-1}\nabla \varphi_id \wtilde_i \right|^2 \right]
\le  \frac{2\sigma}{\beta N} \int_0^T \somma{i}{1}{N-1}(\nabla \varphi_i)^2d t  \nonumber
\le \frac{C T \sigma}{N^2} \nonumber
\end{align}
In order to evaluate the boundary terms in \eqref{eq:Rtilde}, we estimate, for any $i =1$ and $i= N$,
\begin{align}
  \int (V'_i)^2 f_t^N d \lambda_t^N &
= \int \bar \tau(t) V'_i f_t^N d\lambda_t^N+ \int (V'_i-\bar \tau(t)) V'_i f_t^N d\lambda_t^N
\\
                                    & = \int \bar \tau(t) V'_i f_t^N d\lambda_t^N + \int V''(r_i) f_t^N d \lambda_t^N
                                      + \int V'_i \frac{\partial f_t^N}{\partial r_i}d \lambda_t^N \nonumber
\\
                                    & \le C + \alpha \int (V'_i)^2 f_t^N d \lambda_t^N
                                      + \int \left( \frac{\partial f_t^N}{\partial r_i} \right)^2 \frac{1}{f_t^N} d \lambda_t^N,
                                      \nonumber
\end{align}
where we have used Cauchy-Schwartz twice and used the boundedness of $\bar \tau(t)$ and $V''$.
Here we can choose $C>0$ and $0< \alpha<1$ such that do not depend on $t$ or $N$,
so that we have
\begin{equation*}
  \int (V'_i)^2 f_t^N d \lambda_t^N \le C' \left( 1
    + \int \left( \frac{\partial f_t^N}{\partial r_i} \right)^2 \frac{1}{f_t^N} d \lambda_t^N\right)
\end{equation*}
For $i=N$ this gives directly
$(1-\alpha) \int (V'_i)^2 f_t^N d \lambda_t^N \le C' \left( 1 +  \tilde{\mathcal D}_N(t)\right)$.
For $i=1$, by writing
\begin{align}
  \left(\frac{\partial f_t^N}{\partial r_1} \right)^2 &= \left(-\sum_{j=1}^{N-1}
   \left( \frac{\partial f_t^N}{\partial r_{j+1}}- \frac{\partial f_t^N}{\partial r_j} \right)
    + \frac{\partial f_t^N}{\partial r_N} \right)^2
\\
& \le N\left( \sum_{j=1}^{N-1} \left( \frac{\partial f_t^N}{\partial r_{j+1}}- \frac{\partial f_t^N}{\partial r_j} \right)^2 + \left(\frac{\partial f_t^N}{\partial r_N}\right)^2 \right) \nonumber
\end{align}
we obtain
\begin{align}
(1-\alpha) \int (V'_1)^2 f_t^N d \lambda_t^N \le C + N \tilde{\mathcal D}_N(t),
\end{align}
{
and, in turn,
\begin{align}
\int_0^T \int (V'_1)^2 f_t^N d \lambda_t^N  \le C(T) \left( 1 + \frac{N}{\sigma}\right) \le C'(T) \frac N \sigma.
\end{align}
This allows us to estimate
\begin{align}
  \mathbb E \left[ \left(\sigma \int_0^T V'_1(t) \nabla \varphi_0 dt\right)^2 \right] &\le
 T\frac{\sigma^2}{N^2} \mathbb E \left[ \int_0^T (V'_1(t))^2 dt \right]
                                                                                        \le C_T \frac{\sigma}{N},
\end{align}
}
which vanishes as $N \to \infty$. In a similar way we estimate the boundary term involving $V'_N$.

\end{proof}
Thus, we have obtained the following
\begin{prop} \label{prop:eqapprr}
Let $\varphi \in C^2(Q_T)$ such that $\varphi(t,1)=0$ for all $t \in [0,T]$. Then
\begin{equation}
 \frac{1}{N}\somma{i}{1}{N} \varphi_i(T)r_i(T)-\frac{1}{N} \somma{i}{1}{N}\varphi_i(0)r_i(0) - \int_0^T\frac{1}{N}\somma{i}{1}{N}\dot \varphi_i(t) r_i(t) dt + \int_0^T \somma{i}{1}{N-1} p_i\nabla \varphi_i dt \to 0
\end{equation}
in probability as $N \to \infty$
\end{prop}
From similar calculations and recalling that $\psi_0 =0$, we evaluate
\begin{align}
&\frac{1}{N} \somma{i}{1}{N} \psi_i(T)p_i(T) - \frac{1}{N} \somma{i}{1}{N} \psi_i(0) p_i(0) 
\\
&= \int_0^T \frac{1}{N} \somma{i}{1}{N} \dot \psi_i(t) p_i(t)dt - \somma{i}{1}{N-1} V'_i\nabla \psi_i  dt + \int_0^T \psi_N \taubar(t) dt + R_N, \nonumber
\end{align}
where
\begin{align}
R_N=&  \sigma \int_0^T\sum_{i=1}^{N-1}p_i  \Delta \psi_i dt+\sqrt{2 \beta^{-1} \frac{\sigma}{N}} \int_0^T \somma{i}{1}{N-1}\nabla \psi_id w_i+
\\
& -\sigma \int_0^T p_N\nabla \psi_{N-1} dt+\sqrt{2 \beta^{-1} \frac{\sigma}{N}} \int_0^T \varphi_1 d w_0, \nonumber
\end{align}
Similarly to \thref{lem:rtilde}, we prove the following
\begin{lem}
$\mathbb E[| R_N|^2] \to 0$.
\end{lem}
Thus, we have proved
\begin{prop} \label{prop:eqapprp}
Let $\psi \in C^2(Q_T)$ such that $\psi(t,0)= 0$ for all $t \in [0,T]$. Then
\begin{align}
& \frac{1}{N}\somma{i}{1}{N} \psi_i (T)p_i(T)-\frac{1}{N} \somma{i}{1}{N}\psi_i (0)p_i(0) - \int_0^T\frac{1}{N}\somma{i}{1}{N}\dot \psi_i(t) p_i(t) dt +  
\\
&+\int_0^T \somma{i}{1}{N-1} V'_i\nabla \psi_i  dt +\int_0^T \psi(t,1) \taubar(t) dt  \to 0  \nonumber
\end{align}
in probability as $N \to \infty$.
\end{prop}

\subsection{Mesoscopic solutions}

For any sequence $(a_i)_{i \in \mathbb N}$ and any $l \in \mathbb N$,
smoother block averages are defined as
\begin{align}
\hat a_{l,i} := \frac{1}{l} \sum_{|j| < l} \frac{l-|j|}{l} a_{i-j}, \quad i \ge l, 
\end{align}
 where we will choose 
\begin{align}
l = l(N) := \left[ N^\frac{1}{4} \sigma(N)^\frac{1}{2} \right].
\end{align}
\begin{oss}
Since $\sigma/N$ and $N/\sigma^2$ vanish as $N \to \infty$, we may choose $\sigma = N^{\frac 1 2+\alpha}$, with $\alpha \in(0,1/2)$. This means $l$ is of order $N^{\frac 1 2+\frac \alpha 2}$.
\end{oss}
Let $1_{N,i}(x)$ be the indicator function of the interval
$\left[ \dfrac{i}{N}-\dfrac{1}{2N}, \dfrac{i}{N} + \dfrac{1}{2N} \right]$,
then we define functions on $[0,1]$ from the sequence $(a_i)_{i \in \mathbb N}$ as
\begin{align}
\hat a_N(x) := \sum_{i=l+1}^{N-l} 1_{N,i}(x) \hat a_{l,i}, \qquad x \in [0,1].
\end{align}
Notice that for any function $f$ on $\mathbb R$ and $\varphi(x) \in L^1([0,1])$ we have
\begin{align}
  \label{eq:4}
  \int_0^1 f(\hat a_N(x)) \varphi(x) dx  =& \frac 1N \sum_{i=l+1}^{N-l}  f(\hat a_{l,i}) \bar\varphi_i+
  \\ &+f(0) \left(\int_0^{l/N+ 1/(2N)} \varphi(x) dx + \int_{1- l/N -1/(2N)}^1 \varphi(x) dx\right) \nonumber 
\end{align}
where
\begin{equation}\label{eq:barphi}
  {\bar \varphi_i = N \int_0^1 \varphi(x)1_{N,i}(x) dx = N \int_{i/N-1/(2N)}^{i/N+1/(2N)}\varphi(x)dx},
  \qquad i=1,\dots,N-1.
\end{equation}
We have, by Cauchy-Schwarz inequality,
\begin{align}
  \label{eq:5}
  \int_0^1 \hat a_N(x)^2 dx &
 \le \frac 1 N \sum_{k=2}^{N-1} a^2_k \nonumber
\end{align}
Recalling that ${\bf u}_i =(r_i, p_i)$, we define 
\begin{align}
\hat {\bf u}_N(t,x) = \sum_{i=l+1}^{N-l} 1_{N,i}(x) \hat {\bf u}_{l,i}(t), \qquad (t,x) \in Q_T.
\end{align}

As a consequence of \eqref{eq:5} and the energy estimate given by Proposition \ref{prop:energy}, we have that
\begin{equation}
  \label{eq:l2bound}
  \mathbb E\left( \int_{Q_T} |\hat {\bf u}_N(t,x)|^2\right) \le C_T,
\end{equation}
i.e. almost surely $\hat {\bf u}_N(t,x)$ is uniformly bounded in $L^2(Q_T)$ and is therefore weakly convergent, up to a subsequence.
The following proposition ensures us that  ${\pmb \zeta}_N(dx,dt)$ has the same weak limit points as
$\hat {\bf u}_N(t,x)$.
\begin{prop} \label{prop:replace}
  For any function $\varphi \in  C^1([0,1])$ we have that
  \begin{equation}
    \label{eq:6}
    \lim_{N\to\infty} \mathbb E \left(\left| \int_0^1 \hat {\bf u}_N(t,x) \varphi(x)dx -
      \frac 1N \sum_{i=1}^N {\bf u}_i(t) \varphi\left(\frac i N \right)\right| \right) = 0.
  \end{equation}
\end{prop}

  \begin{proof}
  By \eqref{eq:4} with $f(\hat {\bf u}_N)=\hat{ \bf u}_N$ we have
  \begin{align}
  \int_0^1 \hat {\bf u}_N(t,x) \varphi(x)dx = \frac 1 N \sum_{i=l+1}^{N-l} \hat {\bf u}_{l,i} \bar \varphi_i.
  \end{align}
  Next, we note that we can neglect the first and the last $l$ points at the boundaries. Namely, we have
  \begin{align}
\left |  \frac 1 N \sum_{i=1}^l {\bf u}_i(t) \varphi \left( \frac i N \right) \right |& \le \left( \frac 1 l \sum_{i=1}^l \varphi \left( \frac i N \right)^2 \right)^{1/2} \left( \frac 1 l \sum_{i=1}^l \frac{l^2}{N^2} |{\bf u}_i(t)|^2 \right)^{1/2}
\\
& \le C \sqrt \frac l N \left( \frac 1 N \sum_{i=1}^N  |{\bf u}_i(t)|^2 \right)^{1/2},
  \end{align}
so that
\begin{align}
\mathbb E \left[ \left |  \frac 1 N \sum_{i=1}^l {\bf u}_i(t) \varphi \left( \frac i N \right) \right | \right] \to 0.
\end{align}
Similarly, we have
\begin{align}
\mathbb E \left[ \left |  \frac 1 N \sum_{i=N-l+1}^N {\bf u}_i(t) \varphi \left( \frac i N \right) \right | \right] \to 0.
\end{align}
Therefore, we evaluate
\begin{align}
\frac 1 N \sum_{i=l+1}^{N-l} \hat {\bf u}_{l,i} \bar \varphi_i - \frac 1N \sum_{i=l+1}^{N-l} {\bf u}_i \varphi\left(\frac i N \right) &=\frac 1 N \sum_{i=l+1}^{N-l} (\hat {\bf u}_{l,i}-{\bf u}_i) \varphi \left( \frac i N \right)+
\\
&+ \frac 1 N \sum_{i=l+1}^{N-l} \hat {\bf u}_{l,i} \left( \bar \varphi_i -\varphi \left( \frac i N \right) \right). \nonumber
\end{align}
The last summation is estimated by noting that there is a point $\xi_i \in \left[ \dfrac i N- \dfrac 1 {2N}, \dfrac i N+ \dfrac 1 {2N}\right]$ such that $\bar \varphi_i = \varphi(\xi_i)$. Thus,
\begin{align}
\left|\bar \varphi_i -\varphi \left( \frac i N \right)\right| \le \|\varphi'\|_\infty \left| \xi_i-\frac i N\right | \le \frac {\|\varphi'\|_\infty} {2N}
\end{align}
and therefore
\begin{align}
\mathbb E \left[\left |  \frac 1 N \sum_{i=l+1}^{N-l} \hat {\bf u}_{l,i} \left( \bar \varphi_i -\varphi \left( \frac i N \right) \right) \right |\right] &\le  \frac C N \left(\mathbb E \left[  \frac 1 N \sum_{i=l+1}^{N-l} |\hat {\bf u}_{l,i}|^2  \right]\right)^{1/2} \label{eq:523}
\\
& \le  \frac C N \left(\mathbb E \left[  \frac 1 N \sum_{i=1}^N  |{\bf u}_i|^2  \right]\right)^{1/2} \le \frac C N. \nonumber
\end{align}
Finally, defining $c_j = \dfrac{l-|j|}{l^2}$ and recalling that $\sum_{|j|<l}c_j=1$, we write perform a change of variables and write
\begin{align} \label{eq:replace}
\frac 1 N \sum_{i=l+1}^{N-l} (\hat {\bf u}_{l,i}-{\bf u}_i) \varphi \left( \frac i N \right) & =\sum_{|j|<l} c_j \left[\frac 1 N \sum_{i=l+1-j}^{N-l-j} {\bf u}_i \varphi \left( \frac{i+j}{N}\right)-\frac 1 N \sum_{i=l+1}^{N-l} {\bf u}_i \varphi \left( \frac i N \right) \right]
\\
 & =\sum_{|j|<l} c_j \frac 1 N \sum_{i=l+1}^{N-l} {\bf u}_i\left[ \varphi \left( \frac{i+j}{N}\right)- \varphi \left( \frac i N \right) \right] + \mathcal O \left( \sqrt \frac l N \right). \nonumber
\end{align}
The conclusion then follows similarly to \eqref{eq:523}.
\end{proof}
  
The proposition allows us to replace each ${\bf u}_i$ with their average $\hat {\bf u}_{l,i}$. In the same way we can replace $V'(r_i)$ by the average
\begin{align}
\hat V'_{l,i} :=\frac 1 l \sum_{|j|<l} \frac{l-|j|} l V'(r_{i-j})
\end{align}
and then replace $\hat V'_{l,i}$ by $\tau(\hat r_{l,i})$ via the following proposition, which we shall prove in Section \ref{sec:estimates}.
\begin{prop}[One-block estimate] \label{prop:oneblock}
\begin{align}
\lim_{N \to \infty} \mathbb E \left[ \frac 1 N \sum_{i=l+1}^{N-l}  \int_0^T \left( \hat V'_{l,i}- \tau(\hat r_{l,i}) \right)^2 dt \right] = 0.
\end{align}
\end{prop}
Therefore, combining Propositions \ref{prop:eqapprr}, \ref{prop:eqapprp}, \ref{prop:replace} and \ref{prop:oneblock}, we obtain the following
\begin{prop} \label{prop:approssimata}
Let $\varphi, \psi \in C^2(Q_T)$ such that $\varphi(t,1) = \psi(t,0)=0$ for all $t \in [0,T]$ and let ${\bf u}_N(t,x) = (r_N(t,x), p_N(t,x))$. Then, the following convergences happen in probability as $N \to \infty$
\begin{align}
\int_0^1&\left( \varphi(T,x)\hat r_N(T,x)-\varphi(0,x)\hat r_N(0,x) \right) dx +
\\
&- \int_{Q_T} \left(\hat r_N(t,x)\partial_t \varphi(t,x)  + \hat p_N(t,x) \partial_x \varphi(t,x) \right) dx dt \to 0 \nonumber
\end{align}
and
\begin{align}
\int_0^1&\left( \psi(T,x)\hat p_N(T,x)-\psi(0,x)\hat p_N(0,x) \right) dx +
\\
&- \int_{Q_T} \left(\hat p_N(t,x)\partial_t \psi(t,x)  + \tau(\hat r_N(t,x)) \partial_x \psi(t,x) \right) dx dt + \int_0^T \psi(t,1) \bar \tau(t)dt \to 0 \nonumber
\end{align}
\end{prop}

\subsection{Random Young Measures and Weak Convergence}
The purpose of this section is to prove that any weakly convergent subsequence will converge strongly.
We use the compensated compactness argument of Fritz \cite{Fritz1}, inspired from the work of
Di Perna, Serre and Shearer, properly adapted to the presence of boundaries \cite{marchesani2018note}.

Denote by $\nuhat^N = \delta_{\uhat_N(t,x)}$ the random Young measure on $\R^2$ associated to the {empirical} 
 process $\uhat_N(t,x)$:
\begin{equation}
 \int_{\R^2}f(\y)d\nuhat^N(\y) = f(\uhat_N(t,x))
\end{equation}
for any $f : \R^2 \to \R$.  Since $\uhat_N \in L^2(\Omega \times Q_T)$, we say that $\nuhat^N$ is a $L^2$-random Dirac mass. The following
\begin{equation}\label{eq:enyoung1}
{\E \left[ \int_{Q_T} \int_{\R^2}|\y|^2 d \nuhat^N(\y)dxdt \right]= \mathbb E \left[ \int_{Q_T} |\uhat_N(t,x)|^2dxdt \right]} \le C_T
\end{equation}
with $C_T$ independent of $N$, implies that there exists a subsequence of random Young measures $(\nuhat^{N_n})$ and a subsequence of real random variables $(\normp[\uhat_{N_n}]{2})$ that converges in law.

We can now apply the Skorohod's representation theorem to the laws of $( \nuhat^{N_n
},\normp[\uhat_{N_n}]{2})$ and find a common probability space such that the convergence happens almost surely. This proves the following proposition:
\begin{prop}\label{prop:skoro}
There exists a probability space $(\tilde \Omega, \tilde{\mathcal F}, \Pbbtilde)$,  random Young measures $\nutilde^n,\nutilde$ and real random variables $a_n,a$ such that $\nutilde^n$ has the same law of $\nuhat^{N_n}$, $a_n$ has the same law of $\normp[\uhat_{N_n}]{2}$ and $\nutilde^n \weakstar \nutilde$, $a_n \to a$, $\Pbbtilde$-almost surely.
\end{prop}
\begin{oss}
Since $\nuhat^{N_n}$ is a random Dirac mass and $\nutilde^n$ and $\nuhat^{N_n}$ have the same law, $\nutilde^n$ is a $L^2$-random Dirac mass, too: $\nutilde^n = \delta_{\utilde_n(t,x)}$ for some $\utilde_n \in L^2(\tilde \Omega \times Q_T)$.  $\utilde_n$ and $\uhat_{N_n}$ have the same law. Since $a_n \to a$ almost surely, we have that $(a_n)$ is bounded and so $\utilde_n$ is  uniformly bounded in $L^2(Q_T)$ with $\Pbbtilde$-probability $1$. 
\end{oss}
Since from a uniformly bounded sequence in $L^p$ we can extract a weakly convergent subsequence, we obtain the following proposition:
\begin{prop}\label{prop:weakconv}
There exist $L^2(Q_T)$-valued random variables $(\utilde_n), \utilde$ such that $\utilde_n$ and $\uhat_{N_n}$ have the same law and, $\Pbbtilde$-almost surely and up to a subsequence, $\utilde_n \weak \utilde$ in $L^2(Q_T)$.
\end{prop}
The condition $\nutilde^n \weakstar \nutilde$ in Proposition \ref{prop:skoro} reads
\begin{equation} \label{eq:youngf}
\lim_{n\to \infty}  \int_{\R^2} f(\y) d\nutilde^n(\y) dxdt =  \int_{\R^2}f(\y) d\nutilde(\y) dxdt
\end{equation}
for all continuous and \emph{bounded} $f : \mathbb R^2 \to \mathbb R$. By a simple adaptation of Proposition 4.2 of \cite{StochasticYoung}, \eqref{eq:youngf} can be extended to a functiion $f : \mathbb R^2 \to \mathbb R$ such that $f(\y)/|\y|^2 \to 0$ as $|\y| \to +\infty$.

Because of Proposition \ref{prop:approssimata} we are interested in the weak limit of $\tau(\hat r_N(t,x))$. Since $\tau$ is linearly bounded, the main Theorem \ref{thm:main} is proved once we show that $\nutilde = \delta_{\utilde(t,x)}$, almost surely and for almost all $(t,x) \in Q_T$.

We shall now prove  that the support of $\nutilde$ is almost surely and almost everywhere a point. The result will then follow from the lemma:
\begin{lem}\label{lem:lpstrong}
$\nutilde = \delta_{\utilde(t,x)}$ almost surely and for almost all $(t,x) \in Q_T$ if and only if the support of $\nutilde$ is a point for almost all $(t,x) \in Q_T$. In this case, $\utilde_n \to \utilde$ in $L^p(Q_T)^2$-strong for all $1 \le p<2$.
\end{lem}
\begin{proof}
Suppose there is a measurable function ${\bf u}^*: Q_T \to \R^2$ such that $\nutilde = \delta_{{\bf u}^*(t,x)}$ for almost all $(t,x) \in Q_T$. For any test function ${\bf J} : Q_T \to \R^2$ consider the quantity
\begin{equation}
\int_{Q_T} {\bf J}(t,x) \cdot \utilde_n(t,x)dxdt = \int_{Q_T} \int_{\R^2} {\bf J}(t,x) \cdot \y d\nutilde^n(\y)dxdt.
\end{equation}
By taking the limit for $n \to \infty$ {in the sense} of $L^2$-weak first and in the sense of \eqref{eq:youngf} then, we obtain
\begin{equation}
\int_{Q_T} \int_{\R^2} {\bf J}(t,x) \cdot \utilde(t,x)dxdt =\int_{Q_T} \int_{\R^2} {\bf J}(t,x) \cdot \y d\nutilde(\y)dxdt = \int_{Q_T} \int_{\R^2} {\bf J}(t,x) \cdot {\bf u}^*(t,x)dxdt
\end{equation}
almost surely. Then $\utilde(t,x) = {\bf u}^*(t,x)$ for almost all $(t,x) \in Q_T$ follows from the fact that ${\bf J}$ was arbitrary. 

Next, fix $1<p<2$. Taking $f(\y) = |\y|^p$ in \eqref{eq:youngf} gives $\normp[\utilde_n]{p} \to \normp[\utilde]{p}$, which, together with $\utilde_n \weak \utilde$ in $L^p(Q_T)$ and the fact that $L^p(Q_T)$ is uniformly convex for $1<p<\infty$ implies strong convergence.

The case $p=1$ follows from the result for $p>1$ and {H\"{o}lder}'s inequality.
\end{proof}

\subsection{Reduction of the Limit Young Measure}
In this section we prove that the support of $\nutilde$ is almost surely and almost everywhere a point. 

We recall that Lax entropy-entropy flux pair for the system
\begin{equation} \label{eq:psystem1}
\begin{cases}
\partial_t r(t,x)- \partial_x p(t,x) =0
\\
\partial_t p(t,x) - \partial_x\tau(p(t,x))=0
\end{cases}
\end{equation}
is a pair of functions $\eta, q : \R^2 \to \R$ such that
\begin{equation}
\partial_t \eta(\uv(t,x)) + \partial_x q(\uv(t,x)) = 0
\end{equation} 
for any smooth solution $\uv(t,x) = (r(t,x),p(t,x))$ of \eqref{eq:psystem1}. This is equivalent to the following:
\begin{equation}
\begin{cases}
 \partial_r\eta(r,p)+\partial_pq(r,p)=0
\\
\tau'(r)\partial_p\eta(r,p) +   \partial_rq(r,p) =0
\end{cases}.
\end{equation}
Under appropriate conditions on $\tau$, Shearer (\cite{Shearer1}) constructs a family of entropy-entropy flux pairs $(\eta,q)$ such that $\eta,q$, their first and their second derivatives are bounded (cf also \cite{marchesani2018note}). As shown in \cite{marchesani2018hydrodynamic} and \cite{Fritz1}, our choice of the potential $V$ ensures that the tension $\tau$ has the required properties, so the result of Shearer applies to our case.

In particular,  {following Section 5 of }\cite{Shearer1}, we have that the support $\nutilde$ is almost surely and almost everywhere a point provided {Tartar's} 
commutation relation
\begin{equation} \label{eq:tartar}
\expnu[\eta_1q_2-\eta_2q_1] = \expnu[\eta_1]\expnu[q_2] - \expnu[\eta_2]\expnu[q_1]
\end{equation}
holds almost surely and almost everywhere for any bounded pairs $(\eta_1,q_1)$, $(\eta_2,q_2)$ with bounded first and second derivatives.

Obtaining \eqref{eq:tartar} in a deterministic setting is standard and relies on the div-curl and {Murat-Tartar} 
 lemma. Both of these lemmas have a stochastic extension ({cf Appendix A of }\cite{Divcurl}) and what we ultimately need to prove in order to obtain \eqref{eq:tartar} is that the hypotheses for the stochastic Murat-Tartar lemma are satisfied (cf \cite{StochasticYoung} {, Proposition 5.6}). 
 
 This is ensured next theorem, for which we will give a preliminary definition. Let $(\eta, q) \in C^2(\mathbb R^2)$ be a Lax entropy-entropy flux pair with bounded derivatives. We assume, without loss of generality, $\eta(0,0) =q(0,0)=0$.

For $\varphi \in H^1_0(Q_T) \cap L^\infty(Q_T)$ define the corresponding entropy production functional as
\begin{align}
X_N(\varphi, \eta) = - \int_{Q_T} \left(  \eta(\hat {\bf u}_N) \partial_t \varphi+q(\hat {\bf u}_N)  \partial_x \varphi \right)dxdt.
\end{align}
\begin{theorem} \thlabel{thm:laxproduction}
 The entropy production $X_N$ decomposes as $X_N = Y_N + Z_N$, such that $(Y_N)$ is compact in $H^{-1}(Q_T)$ and $(Z_N)$ is uniformly bounded as a signed measure. Namely,
\begin{align}
\mathbb E \left[\left| Y_N(\varphi, \eta)\right| \right] &\le a_N \| \varphi \|_{H^1(Q_T)}\quad \text{with} \quad \lim_{N \to \infty} a_N = 0, \label{eq:typeY}
\\ 
\mathbb E \left[ \left| Z_N(\varphi, \eta)\right| \right] &\le b_N \| \varphi \|_{L^\infty(Q_T)} \quad \text{with} \quad  \limsup_{N \to \infty} b_N < \infty, \label{eq:typeZ}
\end{align}
where $a_N, b_N > 0$ are  independent of $\varphi$.
\end{theorem}
\begin{defn} \thlabel{def:type}
We say that the random variables $Y_N(\varphi, \eta)$ are of \emph{type Y} provided \eqref{eq:typeY} holds for some $a_N$ independent of $\varphi$. We further say that the random variables $Z_N(\varphi, \eta)$ are of \emph{type Z} provided \eqref{eq:typeZ} holds for some $b_N$ independent of $\varphi$.
\end{defn}
By recalling that $\varphi$ vanishes on $\partial Q_T$, a direct calculation involving Ito formula we can integrate by parts in time and obtain
\begin{equation} \label{eq:entropyDec}
X_N = X_{a,N} + X_{s,N}+\tilde X_{s,N} + \mathcal M_N+\tilde{\mathcal M}_N+ Q_N ,
\end{equation}
where
\begin{align} \label{eq:XaN}
  X_{a,N}(\varphi,\eta) =&
      \int_0^T   \sum_{i=l+1}^{N-l} \bar \varphi_i \left[  \partial_p\eta(\uhat_{l,i}) \nabla \hat V'_{l,i}
      - \partial_r\eta(\uhat_{l,i}) \nabla^*\hat p_{l,i} \right] dt +
\\
     &+ \int_0^T  \sum_{i=l+1}^{N-l} \bar \varphi_i \left[ \partial_r q(\uhat_{l,i}) \nabla \hat r_{l,i}
       - \partial_pq(\uhat_{l,i})\nabla^* \hat p_{l,i} \right] dt
       , \nonumber
\end{align}
\begin{align} \label{eq:XsN}
X_{s,N}(\varphi,\eta)= \sigma \int_0^T \sum_{i=l+1}^{N-l} \bar \varphi_i    \partial_p\eta(\uhat_{l,i}) \Delta \hat p_{l,i}  dt+ \frac{2\sigma}{l^3} \int_0^T \sum_{i=l+1}^{N-l}  \bar \varphi_i \partial^2_{pp}\eta(\uhat_{l,i}) dt ,
\end{align}
\begin{align} \label{eq:tilXsN}
\tilde X_{s,N}(\varphi,\eta)= \sigma \int_0^T  \sum_{i=l+1}^{N-l} \bar \varphi_i   \partial_r\eta(\uhat_{l,i}) \Delta \hat V'_{l,i}dt+\frac{2\sigma}{l^3} \int_0^T  \sum_{i=l+1}^{N-l} \bar \varphi_i \partial^2_{rr}\eta(\uhat_{l,i}) dt,
\end{align}
\begin{align} \label{eq:MN}
 M_N(\varphi, \eta) = - \sqrt{2\frac{\sigma}{N}} \int_0^T   \sum_{i=l+1}^{N-l}  \bar \varphi_i    \partial_p\eta(\uhat_{l,i}) d \nabla^* \hat w_{l,i},
\end{align}
\begin{align} \label{eq:tildeMN}
\tilde{ M}_N(\varphi, \eta) = - \sqrt{2\frac{\sigma}{N}} \int_0^T   \sum_{i=l+1}^{N-l}  \bar \varphi_i    \partial_r\eta(\uhat_{l,i})   d \nabla^*\hat{\tilde w}_{l,i} 
\end{align}
\begin{align} \label{eq:QN}
 Q_N(\varphi,\eta) &= -\int_0^T \int_0^1  \partial_x\varphi(t,x) q(\uhat_N(t,x)) dx d t
\\
&- \int_0^T  \sum_{i=l+1}^{N-l} \bar \varphi_i \left[ \partial_r q(\uhat_{l,i}) \nabla \hat r_{l,i} - \partial_pq(\uhat_{l,i})\nabla^* \hat p_{l,i} \right] dt. \nonumber
\end{align}
We shall prove \thref{thm:laxproduction} via a series of lemmas. We start with two preliminary ones.
\begin{lem} \thlabel{lem:typeZ}
Let $(A_i)_{i \in \mathbb N}$ and $(B_i)_{i \in \mathbb N}$ be families of $L^2(\mathbb R)$-valued random variables such that
\begin{equation}
\limsup_{N \to \infty}\left( \mathbb E \left[\sum_{i=1}^N \int_0^T A_i(s)^2ds \right]\mathbb E \left[\sum_{i=1}^N \int_0^T B_i(s)^2ds \right]\right) < \infty.
\end{equation}
 Let $\varphi \in {L^\infty(Q_T)}$  and let
 \begin{align} \label{eq:barphi2}
 \bar \varphi_i(t) := N\int_0^1 1_{N,i}(x) \varphi(t,x) dx = N\int_{\frac i N- \frac 1 {2N}}^{\frac i N+\frac 1 {2N}} \varphi(t,x)dx
 \end{align}
 Then 
\begin{equation}
\mathbb E \left[\left |\sum_{i=1}^N \int_0^T \bar \varphi_i A_i B_i dt \right | \right] \le   b_N \| \varphi \|_{L^\infty(Q_T)},
\end{equation}
where $ b_N$ is independent of $\varphi$ such that
\begin{equation}
\limsup_{N \to \infty} b_N  < \infty.
\end{equation}
\end{lem}
\begin{proof}
\begin{equation}
|\bar \varphi_i| = \left|N \int_0^1 \varphi(t,x)1_{N,i}(x) dx \right| \le \| \varphi \|_{L^\infty(Q_T)}.
\end{equation}
Then, by the Cauchy-Schwarz inequality, we have
\begin{equation}
\mathbb E \left[\left |\sum_{i=1}^N \int_0^T \bar \varphi_i A_i B_i dt\right | \right] \le \|\varphi\|_\infty{  \left( \mathbb E\left[\sum_{i=1}^N \int_0^T A_i^2dt \right]\right)^{1/2} \left(  \mathbb E \left[\sum_{i=1}^N \int_0^T B_i^2dt \right]\right)^{1/2}}.
\end{equation}
\end{proof}
\begin{lem} \thlabel{lem:typeY}
Let $(A_i)_{i \in \mathbb N}$ be a family of $L^2(\mathbb R)$-valued random variables such that
\begin{equation}
\limN\mathbb E \left[ \frac{1}{N}\sum_{i=1}^N \int_0^T A_i(s)^2ds \right] =0.
\end{equation}
Let $\varphi \in H^1(Q_T)$ and $\bar \varphi_i$ as in \eqref{eq:barphi2}. Then
\begin{equation}
\mathbb E \left[ \sum_{i=1}^{N-1} \int_0^T  \left| A_i \left( \bar \varphi_{i+1}-\bar \varphi_i\right) \right| dt\right] \le a_N \|\varphi \|_{H^1(Q_T)},
\end{equation}
where $a_N$ is  independent of $\varphi$  and
\begin{equation}
\lim_{N \to \infty} a_N = 0.
\end{equation}
\end{lem}
\begin{proof}
By Cauchy-Schwarz we have
\begin{align} \label{eq:typeYCS}
\mathbb E \left[\sum_{i=1}^{N-1} \int_0^T  \left|A_i\left( \bar \varphi_{i+1}-\bar \varphi_i\right) \right|dt \right] \le& \left(\frac{1}{N}\sum_{i=1}^{N-1}\int_0^T N^2 \left( \bar \varphi_{i+1}-\bar \varphi_i\right)^2 dt \right)^{1/2} \times
\\
&\times \left( \mathbb E \left[\frac{1}{N}\sum_{i=1}^{N-1}\int_0^T A_i^2 dt \right]\right)^{1/2}. \nonumber
\end{align}
We write
\begin{align}
\bar\varphi_{i+1} -\bar \varphi_i &=N \int_0^1 1_{N,i+1}(x) \varphi(t,x)dx- N\int_0^1 1_{N,i}(x) \varphi(t,x)dx
\\
&= N\int_0^1 1_{N,i}(x) \left( \varphi \left(t, x+\frac{1}{N}\right)-\varphi(t,x)\right)dx \nonumber
\\
&=  {N\int_0^1 1_{N,i}(x) \int_x^{x+\frac{1}{N}}\partial_y \varphi(t,y)dydx} \nonumber
\\
&= N \int_0^1 1_{N,i}(x) \int_0^1 1_{\left[x,x+\frac{1}{N}\right]}(y) \partial_y \varphi(t,y) dy dx \nonumber
\end{align}
 {Thus, Cauchy-Schwarz inequality implies}
\begin{align}
 | \bar \varphi_{i+1}-\bar \varphi_i | & \le N \int_0^1 1_{N,i}(x) \left| \int_0^1 1_{\left[x,x+\frac{1}{N}\right]}(y) \partial_y \varphi(t,y) dy \right | dx
 \\
 & \le  N \int_0^1 1_{N,i}(x)  \left( \int_0^1 1_{\left[x,x+\frac{1}{N}\right]}(y) dy \right)^{1/2} \left( \int_0^1|\partial_y \varphi(t,y)|^2 dy \right)^{1/2} dx \nonumber
 \\
 & = \frac{1}{\sqrt N} \left( \int_0^1|\partial_y \varphi(t,y)|^2 dy \right)^{1/2} \nonumber
\end{align}
 {and so}
\begin{align} \label{eq:gradphi}
 \int_0^T\frac{1}{N} \sum_{i=1}^{N-1}N^2\left( \bar \varphi_{i+1}-\bar \varphi_i\right)^2 dt & \le   \int_0^T \int_0^1| \partial_x\varphi(t,x)|^2 dxdt 
 \\
& = \|\partial_x\varphi\|_{L^2(Q_T)}^2 \le  \|\varphi\|_{H^1(Q_T)}^2. \nonumber
\end{align}
\end{proof}
\begin{oss}
The same result applies if we replace $\nabla \bar \varphi_i$ by $\nabla^* \bar \varphi_i$.
\end{oss}
In the following we shall diffusely use the following formulae, which hold for any two sequences $(a_i)_{i \in \mathbb N}$, $(b_i)_{i \in \mathbb N}$.
\begin{align}
\nabla (a_ib_i) &= b_{i+1} \nabla a_i + a_i \nabla b_i
\\
\nabla^*(a_ib_i) &= b_{i-1}\nabla^* a_i + a_i \nabla^* b_i
\\
\sum_{i=l+1}^{N-l} a_i \nabla b_i &= \sum_{i=l+1}^{N-l} b_i \nabla^* a_i + a_{N-l}b_{N-l+1}-a_lb_{l+1}
\end{align}
\begin{lem}
$X_{a,N}$ be defined in \eqref{eq:XaN}. Then it decomposes as $X_{a,N}= Y_{a,N}+Z_{a,N}$ such that $Y_{a,N}$ is of type Y and $Z_{a,N}$ is of type $Z$ in the sense of \thref{def:type}. Moreover, $b_N \to 0$ as $N \to \infty$.
\end{lem}
\begin{proof}
Since $(\eta,q)$ is a Lax entropy-entropy flux pair, we have
\begin{align}
\begin{cases}
\partial_r \eta + \partial_p q =0
\\
\tau'(r)\partial_p \eta +  \partial_r q =0
\end{cases},
\end{align}
so that we can write
\begin{align}
X_{a,N}&= \int_0^T \sum_{i=l+1}^{N-l} \bar \varphi_i \partial_p \eta(\hat u_{l,i})\left[\nabla \hat V'_{l,i}- \tau'(\hat r_{l,i}) \nabla \hat r_{l,i}\right]dt
\\
& = X^1_{a,N}+ Z^1_{a,N},
\end{align}
where
\begin{align}
X^1_{a,N} & = \int_0^T \sum_{i=l+1}^{N-l} \bar \varphi_i \partial_p \eta(\hat u_{l,i})\nabla \left( \hat V'_{l,i}- \tau(\hat r_{l,i}) \right)dt
\\
Z^1_{a,N} & = \int_0^T \sum_{i=l+1}^{N-l} \bar \varphi_i \partial_p \eta(\hat u_{l,i})\left[\nabla \tau(\hat r_{l,i})- \tau'(\hat r_{l,i}) \nabla \hat r_{l,i}\right]dt
\end{align}
The term $Z^1_{a,N}$ is of type Z with $b_N \to 0$. This follows from \thref{lem:typeZ}, the fact that $\partial_p\eta$ and $\tau''$ are bounded \thref{cor:2block}, and the following estimate,
\begin{align}
|\nabla \tau(\hat r_{l,i})-\tau'(\hat r_{l,i}) \nabla \hat r_{l,i} | &\le |[\tau'(\alpha \hat r_{l,i}+(1-\alpha)\hat r_{l,i+1})-\tau'(\hat r_{l,i})] \nabla \hat r_{l,i}| 
\\
&\le \| \tau'' \|_\infty(1-\alpha) (\nabla r_{l,i})^2, \nonumber
\end{align}
which holds for some  $\alpha \in (0,1)$.

After a summation by parts, we write
\begin{align}
X^1_{a,N} = Y_{a,N} + Z^2_{a,N} + Z^b_{a,N},
\end{align}
where
\begin{align}
Y_{a,N} =& \int_0^T \sum_{i=l+1}^{N-l} (\nabla^* \bar \varphi_i) \partial_p \eta(\hat u_{l,i-1}) (\hat V'_{l,i}-\tau(\hat r_{l,i}))dt
\\
Z^2_{a,N} = &\int_0^T \sum_{i=l+1}^{N-l}  \bar \varphi_i (\nabla^* \partial_p \eta(\hat u_{l,i})) (\hat V'_{l,i}-\tau(\hat r_{l,i}))dt
\\
Z^b_{a,N}  =& \int_0^T \bar \varphi_{N-l} \partial_p \eta(\hat u_{l,N-l}) ( \hat V'_{l,N-l+1}-\tau(\hat r_{l,N-l+1})dt +
\\
&- \int_0^T\bar \varphi_l \partial_p \eta(\hat u_{l,l})(\hat V'_{l,l+1}-\tau(\hat r_{l,l+1}))dt. \nonumber
\end{align}
$Y_{a,N}$ is of type Y. This follows from \thref{lem:typeY}, the fact that $\partial_p \eta$ is bounded and \thref{cor:1block}. $Z^2_{a,N}$ is of type Z with $b_N \to 0$. This follows from \thref{lem:typeZ}, \thref{cor:1block} and \thref{cor:2block} after writing, for some intermediate $\tilde u_{l,i}$,
\begin{align}
\nabla^* \partial_p \eta(\hat u_{l,i}) =  \partial^2_{rp} \eta(\tilde u_{l,i}) \nabla \hat r_{l,i}+\partial^2_{pp}\eta(\tilde u_{l,i}) \nabla \hat p_{l,i}
\end{align}
and using the fact that the second derivatives of $\eta$ are bounded. Finally, $Z^b_{a,N}$ is of type Z with $b_N \to 0$.This follows from
\begin{align}
|Z^b_{a,N}| \le \|\varphi \|_{L^\infty(Q_T)} \| \partial_p \eta \|_\infty \int_0^T\left( | \hat V'_{l,N-l+1}-\tau(\hat r_{l,N-l+1})|+|\hat V'_{l,l+1}-\tau(\hat r_{l,l+1})|\right)dt
\end{align}
and \thref{cor:1block}.
\end{proof}
\begin{lem} \thlabel{lem:XsN}
Let $X_{s,N}$ be defined in \eqref{eq:XsN}. Then it decomposes as $X_{s,N}= Y_{s,N}+Z_{s,N}$ such that $Y_{s,N}$ is of type Y and $Z_{s,N}$ is of type Z in the sense of \thref{def:type}.
\end{lem}
\begin{proof}
Write
\begin{align}
X_{s,N} = X^1_{s,N}+Z^1_{s,N}
\end{align}
where
\begin{align}
X^1_{s,N} &=\sigma \int_0^T \sum_{i=l+1}^{N-l} \bar \varphi_i    \partial_p\eta(\uhat_{l,i}) \Delta \hat p_{l,i}  dt
\\
Z^1_{s,N}&= \frac{2\sigma}{l^3} \int_0^T \sum_{i=l+1}^{N-l}  \bar \varphi_i \partial^2_{pp}\eta(\uhat_{l,i}) dt
\end{align}
Since $\partial^2_{pp}\eta$ is bounded, we have that $Z^1_{s,N}$ is of type Z with $b_N \to 0$. In order to evaluate $X^1_{s,N}$ we sum by parts and write
\begin{align}
X^1_{s,N} = Y^1_{s,N} + Z^2_{s,N}+Y^b_{s,N},
\end{align}
where
\begin{align}
Y^1_{s,N}& = - \sigma \int_0^T \sum_{i=l+1}^{N-l} (\nabla \bar \varphi_i)    \partial_p\eta(\uhat_{l,i+1}) \nabla \hat p_{l,i}  dt
\\
Z^2_{s,N}& = - \sigma \int_0^T \sum_{i=l+1}^{N-l} \bar \varphi_i (\nabla  \partial_p\eta(\uhat_{l,i})) \nabla \hat p_{l,i}  dt
\\
Y^b_{s,N} &=  \sigma \int_0^T\bar \varphi_{N-l} \partial_p \eta(\hat u_{l,N-l}) \nabla \hat p_{l,N-l+1}dt-\sigma \int_0^T\bar \varphi_l \partial_p\eta(\hat u_{l,i})\nabla \hat p_{l,l+1}dt.
\end{align}
Thanks to \thref{lem:typeY}, \thref{lem:typeZ} and \thref{cor:2block}, $Y^1_{s,N}$ is of type Y, and $Z^2_{s,N}$ is of type Z. However, we obtain $|Z^2_{s,N}| \le b_N \|\varphi\|_\infty$ with $b_N$ only bounded, since
\begin{align}
\mathbb E \left[  \sum_{i=l+1}^{N-l} \int_0^T (\sigma(\nabla \hat p_{l,i})^2+\sigma(\nabla \hat r_{l,i})^2) dt\right]
\end{align} 
is bounded but does not necessarily vanish as $N \to \infty$. $Y^b_{s,N}$ is of type Y. We focus only on the boundary term in $l$, as the boundary term in $N-l$ is analogous. Since $\varphi(t,0)=0$ we write
\begin{align}
\bar \varphi_l = N \int_0^1 1_{N,l}(x) \varphi(x)dx &= N \int_0^1 1_{N,l}(x) \int_0^x \partial_y \varphi(t,y)dy 
\\
&=N \int_0^1 1_{N,l}(x) \int_0^1 1_{[0,x]}(y) \partial_y \varphi(t,y)dy. \nonumber
\end{align}
By Cauchy-Schwarz, we estimate
\begin{align}
\left | \int_0^1 1_{[0,x]}(y) \partial_y \varphi(t,y)dy \right | \le \sqrt x \left( \int_0^1 | \partial_y \varphi(t,y)|^2 dy \right)^{1/2},
\end{align}
so that
\begin{align}
|\bar  \varphi_l| &\le  \left( \int_0^1 | \partial_y \varphi(t,y)|^2 dy \right)^{1/2} N \int_{\frac{l}{N}-\frac{1}{2N}}^{\frac{l}{N}+\frac{1}{2N}}\sqrt x dx 
\\
&=  \left( \int_0^1 | \partial_y \varphi(t,y)|^2 dy \right)^{1/2} \sqrt \frac{l}{N}  \left( 1+ \mathcal O \left( \frac{l}{N}\right) \right). \nonumber
\end{align}
Thus, we obtain
\begin{align}
\mathbb E \left[\sigma \int_0^T|\bar \varphi_l \partial_p\eta(\hat u_{l,i})\nabla \hat p_{l,l+1}|dt \right]& \le  \|\partial_p \eta\|_\infty \left( \int_0^T |\bar \varphi_l|^2 dt \right)^{1/2} \left( \mathbb E \left[\int_0^T \sigma^2 (\nabla \hat p_{l,i})^2 dt\right] \right)^{1/2} \nonumber
\\
& \le C \left( \int_0^T \int_0^1 |\partial_y \varphi(t,y)|^2 dy \right)^{1/2} \left( \frac{l}{N} \frac{\sigma^2}{l\sigma}\right)^{1/2} 
\\
& \le C \|\varphi \|_{H^1(Q_T)} \left( \frac{\sigma}{N}\right)^{1/2} \nonumber
\end{align}
\end{proof}
Similarly, we have the following
\begin{lem}
Let $\tilde X_{s,N}$ be defined in \eqref{eq:tilXsN}. Then it decomposes as $\tilde X_{s,N}= \tilde Y_{s,N}+\tilde Z_{s,N}$ such that $\tilde Y_{s,N}$ is of type Y and $\tilde Z_{s,N}$ is of type Z in the sense of \thref{def:type}.
\end{lem}
\begin{lem}
Let $ M_N$ be defined in \eqref{eq:MN}. Then it decomposes as $ M_N= Y_{m,N}+Z_{m,N}$ such that $Y_{m,N}$ is of type Y and $Z_{m,N}$ is of type Z in the sense of \thref{def:type}. Moreover, $b_N \to 0$ as $N \to \infty$. 
\end{lem}
\begin{proof}
By a summation by parts we obtain
\begin{align}
M_N =  Y_{m,N}+ Z_{m,N}+Z^b_{m,N},
\end{align}
with
\begin{align}
Y_{m,N} &=- \sqrt{2 \frac{\sigma}{N}} \int_0^T  \sum_{i=l+1}^{N-l} (\nabla \bar\varphi_i) \partial_p \eta(\hat u_{l,i+1}) d \hat w_{l,i}
\\
Z_{m,N}&= - \sqrt{2 \frac{\sigma}{N}} \int_0^T  \sum_{i=l+1}^{N-l}  \bar\varphi_i (\nabla \partial_p \eta(\hat u_{l,i})) d \hat w_{l,i}
\\
Z^b_{m,N} &=\sqrt{2 \frac{\sigma}{N}} \int_0^T \left( \bar \varphi_{N-l+1} \partial_p \eta(\hat u_{N-l+1}) d \hat w_{l,N-l}- \bar \varphi_{l+1}\partial_p\eta(\hat u_{l,l+1}) d \hat w_{l,l} \right).
\end{align}
$Y_{m,N}$ is of type Y. In fact, we have
\begin{align}
\mathbb E \left[Y_{m,N}^2\right] & =   {2 \frac{\sigma}{N}}\mathbb E \left[ \left(\frac{1}{l}\sum_{|j|<l}  \int_0^T\sum_{i=l+1}^{N-l}  (\nabla \bar \varphi_i) \partial_p \eta(\hat u_{l,i+1}) d w_{i-j}\right)^2\right]
\\
&\le  { \frac{2\sigma}{N}} \frac{1}{l}\sum_{|j|<l} \frac{l-|j|}{l}  \sum_{i=l+1}^{N-l} \mathbb E \left [ \int_0^T \left(  (\nabla \bar \varphi_i) \partial_p \eta(\hat u_{l,i+1}) \right)^2 \right ] dt  \nonumber
\\
& \le  { \frac{2\sigma}{N}} \|\partial_p \eta \|_\infty^2 \frac{1}{N^2} \int_0^T \sum_{i=l+1}^{N-l} N^2 (\nabla \bar \varphi_i)^2 dt  \nonumber
\\
& \le \|\varphi\|_{H^1(Q_T)}^2{ \frac{2\sigma}{N^2}} \|\partial_p \eta \|_\infty^2. \nonumber
\end{align}
Thanks to the coefficient $\sqrt{\sigma/N}$, the boundary term $Z^b_{m,N}$ is of type Z with $b_N \to 0$. Finally, we show that $Z_{m,N}$ is of type Z as well, by estimating
\begin{align}
\mathbb E \left[ Z_{m,N}^2 \right]& 
 \le C \|\varphi\|_{L^\infty(Q_T)}^2  \frac{\sigma}{N} \int_0^T \sum_{i=l+1}^{N-l}    \mathbb E \left[  (\nabla \hat r_{l,i})^2+(\nabla \hat p_{l,i})^2 \right] dt \nonumber
\\
& \le  \|\varphi\|_{L^\infty(Q_T)}^2  \frac{C}{N} \nonumber
\end{align}
\end{proof}
Similarly, we prove the following
\begin{lem}
Let $\tilde{ M}_N$ be defined in \eqref{eq:tildeMN}. Then it decomposes as $ M_N= \tilde Y_{ m,N}+\tilde Z_{ m,N}$ such that $\tilde Y_{m,N}$ is of type Y and $\tilde Z_{m,N}$ is of type Z in the sense of \thref{def:type}. Moreover, $b_N \to 0$ as $N \to \infty$.
\end{lem}
\begin{lem}
Let $Q_N$ be defined in \eqref{eq:QN}. Then it decomposes as $Q_N= Y_{q,N}+Z_{q,N}$ such that $Y_{q,N}$ is of type Y and $Z_{q,N}$ is of type Z in the sense of \thref{def:type}. Moreover, $b_N \to 0$ as $N \to \infty$.
\end{lem}
\begin{proof}
Let $\tilde \varphi_i(t) = \varphi \left(t, \dfrac{i}{N}+\dfrac{1}{2N} \right)$. Then
\begin{align}
- \int_0^T \int_0^1 \partial_x \varphi(t,x) q(\hat u_N(t,x)) dx dt =&- \int_0^T\sum_{i=l+1}^{N-l} \left( \int_0^1 \partial_x \varphi(t,x) 1_{N,i}(x) dx \right) q(\hat u_{l,i}) dt \nonumber
\\
 =& \int_0^T \sum_{i=l+1}^{N-l} (\nabla^* \tilde \varphi_i) q(\hat u_{l,i}) dt  
\\
 =&\int_0^T \sum_{i=l+1}^{N-l}  \tilde \varphi_i \nabla  q(\hat u_{l,i}) dt + \int_0^T \bigl(\tilde \varphi_l q(\hat u_{l,l+1})+ \nonumber
 \\
& -\tilde \varphi_{N-l}q(\hat u_{N-l+1}) \bigr)dt. \nonumber
\end{align}
Thus, since we may write, for some $\alpha \in (0,1)$,
\begin{align}
\nabla q(\hat u_{l,i}) = \partial_r \eta(\hat u^\alpha_{l,i}) \nabla \hat r_{l,i}+\partial_p\eta(\hat u^\alpha_{l,i}) \nabla \hat p_{l,i},
\end{align}
where $\hat u^\alpha_{l,i} = \alpha \hat u_{l,i}+(1-\alpha)\hat u_{l,i+1}$, we obtain
\begin{align}
Q_N = Q^r_{q,N}+Q^p_{q,N}+ Z^b_{q,N},
\end{align}
where
\begin{align}
Q^r_{q,N} &= \int_0^T \sum_{i=l+1}^{N-l} \left( \tilde \varphi_i \partial_r q(\hat u^\alpha_{l,i})\nabla \hat r_{l,i}- \bar \varphi_i \partial_r q(\hat u_{l,i}) \nabla \hat r_{l,i} \right)dt
\\
Q^{*p}_{q,N} &=\int_0^T \sum_{i=l+1}^{N-l} \left( \tilde \varphi_i \partial_p q(\hat u^\alpha_{l,i})\nabla \hat p_{l,i}+ \bar \varphi_i \partial_p q(\hat u_{l,i}) \nabla^* p_{l,i} \right)dt
\\
Z^b_{q,N} &=\int_0^T \left(\tilde \varphi_l q(\hat u_{l,l+1})-\tilde \varphi_{N-l}q(\hat u_{N-l+1}) \right)dt.
\end{align}
Since $q$ is bounded, $Z^b_{q,N}$ is at glance of type Z. Next, we write
\begin{align}
Q^r_{q,N} = Y^r_{q,N}+Z^r_{q,N},
\end{align}
where
\begin{align}
Y^r_{q,N} &= \int_0^T \sum_{i=l+1}^{N-l} ( \tilde \varphi_i-\bar \varphi_i)  \partial_r q(\hat u^\alpha_{l,i})\nabla \hat r_{l,i}dt
\\
Z^r_{q,N} &=\int_0^T \sum_{i=l+1}^{N-l} \bar \varphi_i \left(\partial_r q(\hat u^\alpha_{l,i})- \partial_r q(\hat u_{l,i})\right)\nabla \hat r_{l,i}dt.
\end{align}
In order to estimate $Y^r_{q,N}$ we estimate
\begin{align}
|\tilde \varphi_i - \bar \varphi_i |&=N \left|\int_0^1  1_{N,i}(x)\left( \varphi \left( \frac{i}{N}+\frac{1}{2N}\right) - \varphi(x) \right) dx \right|
\\
& \le N \int_0^1  1_{N,i}(x)\left | \int_0^1 1_{[x,\frac{i}{N}+\frac{1}{2N}]}(y) \partial_y\varphi(t,y)dy \right| dx \nonumber
\\
& \le N \int_0^1 1_{N,i}(x) \left( \frac{i}{N}+\frac{1}{2N} -x \right)^{1/2} \left( \int_0^1 |\partial_y\varphi(t,y)|^2 dy \right)^{1/2}dx \nonumber
\\
& = \frac{2}{3\sqrt N} \left( \int_0^1 |\partial_y\varphi(t,y)|^2 dy \right)^{1/2}. \nonumber
\end{align}
Thus, by Cauchy-Schwarz,
\begin{align}
|Y^r_{q,N}|& \le C \left( \int_0^T \sum_{i=l+1}^{N-l} (\tilde \varphi_i- \bar \varphi_i)^2 dt \right)^{1/2} \left( \int_0^T \sum_{i=l+1}^{N-l} (\nabla \hat r_{l,i})^2 dt \right)^{1/2}
\\
& \le C  \left( \int_0^T \int_0^1 |\partial_y\varphi(t,y)|^2dy dt \right)^{1/2} \left( \int_0^T \sum_{i=l+1}^{N-l} (\nabla \hat r_{l,i})^2 dt \right)^{1/2}  \nonumber
\\
& \le C \|\varphi \|_{H^1(Q_T)}\left( \int_0^T \sum_{i=l+1}^{N-l} (\nabla \hat r_{l,i})^2 dt \right)^{1/2}  \nonumber
\end{align}
and so $Y^r_{q,N}$ is of type Y. $Z^r_{q,N}$ is easily seen to be of type Z with $b_N \to 0$ since we can estimate, for some intermediate value $\tilde u_{l,i}$,
\begin{align}
\left| \partial_r q(\hat u^\alpha_{l,i}) - \partial_r q(\hat u_{l,i}) \right| & = |\partial_r q( \hat \alpha u_{l,i}+(1-\alpha) \hat u_{l,i+1})-\partial_rq(\hat u_{l,i})|
\\
& = | \partial^2_{rr} q(\tilde u_{l,i})(1-\alpha) \nabla \hat r_{l,i}+\partial^2_{rp} q(\tilde u_{l,i})(1-\alpha) \nabla \hat p_{l,i}| \nonumber
\\
& \le \|q''\|_\infty (1-\alpha) \left(|\nabla \hat r_{l,i}|+|\nabla \hat p_{l,i}| \right). \nonumber
\end{align}
In order to evaluate $Q^{*p}_{q,N}$ we write
\begin{align}
 Q^{*p}_{q,N} =  Q^p_{q,N}+ X_{q,N},
\end{align}
where
\begin{align}
 Q^p_{q,N} &= \int_0^T \sum_{i=l+1}^{N-l} \left( \tilde \varphi_i \partial_p q(\hat u^\alpha_{l,i})\nabla \hat p_{l,i}- \bar \varphi_i \partial_p q(\hat u_{l,i}) \nabla p_{l,i} \right)dt
\\
X_{q,N} &= \int_0^T \sum_{i=p+1}^{N-l} \bar \varphi_i \partial_p q(\hat u_{l,i}) \left( \nabla^*\hat p_{l,i}+\nabla \hat p_{l,i} \right) dt.
\end{align}
$ Q^p_{p,N}$ is evaluated exactly as $Q^r_{p,N}$ and so it writes as a sum of a type Y term and of a type Z term with $b_N \to \infty$. Finally, since
\begin{align}
\nabla^* \hat p_{l,i}+\nabla \hat p_{l,i} = \Delta \hat p_{l,i},
\end{align}
the term $X_{q,N}$ is entirely similar to the term $X^1_{s,N}$ of \thref{lem:XsN}, except  the latter has an extra factor $\sigma$. Thus, $X_{q,N}$ writes as a sum of a type Y term and a type Z term with $b_N \to 0$.
\end{proof}

\section{Clausius Inequality} \label{sec:clausius-inequality}
This section is devoted to proving the second law of Thermodynamics in the form of the Clausius inequality.  We recall here the variational formula for the relative entropy
\begin{align}\label{eq:varentr}
H_N(t) = \sup_\phi \left \{ \int \phi d \mu_t^N - \log \int \e \phi d \lambda_t^N \right\},
\end{align}
where the supremum is carried over all
measurable function $\phi$ such that  $\int \e \phi d \lambda_t^N < +\infty$.
\begin{lem}
Any solution $\tilde {\bf u}$ belongs almost surely to $L^\infty(0,T;L^2(0,1))$.
\end{lem}
\begin{proof}
Fix $1 \le p <2$. By Lemma \ref{lem:lpstrong} there exists a set $A$ of probability $1$ such that $\tilde {\bf u}_n \to \tilde{\bf u}$ in $L^p$-strong for for any $\omega \in A$. For any such $\omega$ we can find a subsequence $\{n_k^\omega\}$ such that $\tilde{\bf u}_{n_k^\omega}(t,x) \to \tilde {\bf u}(t,x)$ for almost all $t$ and $x$. In particular, for almost all $t$, the sequence $\tilde{\bf u}_{n_k^{\omega}}(t,x)$ converges for almost all $x$. Thus, by Fatou lemma and the remark following Proposition \ref{prop:skoro} we have
\begin{align}
\int_0^1 |\tilde{\bf u}(t,x)|^2 dx \le \liminf_{k\to \infty} \int_0^1 |\tilde {\bf u}_{n_k^\omega}(t,x)|^2 dx \le C
\end{align}
for almost all $t$.
\end{proof}
\begin{lem}
  For any limit point $\mathfrak Q$ of $\mathfrak Q_N$ 
  \begin{equation}
\int_0^T    \liminf_{N\to\infty}\frac{H_N(t)}{N} dt\ \ge\ \beta
    \mathbb E^{\mathfrak Q}\left(  \int_0^T\int_0^1 \left[\mathcal F(t,y) - \bar\tau(t) r(t,y)\right] dy dt \right) + \int_0^TG(\bar\tau(t))dt,
    \label{eq:entropy-macro}
  \end{equation}
  where
  \begin{equation}
    \label{eq:14}
    \mathcal F(t,y) := \frac{p(t,y)^2}{2} + F( r(t,y)). 
  \end{equation}
\end{lem}

  \begin{proof}

    To get the lower bound for any fixed $t \ge 0$ choose two Lipschitz functions $\tilde \tau(t,\cdot)$ and $\tilde p(t,\cdot)$ on $[0,1]$ and define $\widetilde{\tau_i(t)} := \tilde\tau\left(t, \dfrac i N\right)$
    and $\widetilde{p_i(t)} := \tilde p\left(t, \dfrac i N \right)$. Then take $\phi$ in \eqref{eq:varentr} as
\begin{equation}
  \label{eq:13}
  \phi({\bf r}, {\bf p}) =  \sum_{i=1}^N \left\{\beta\left[(\widetilde{\tau_i(t)}- \bar\tau(t)) r_i +\widetilde{p_i(t)} p_i\right] -
  G( \widetilde{\tau_i(t)}) +G(\bar\tau(t)) - \beta\frac{\widetilde{p_i(t)}^2}{2} \right\}.
\end{equation}
Observe that $\phi$ is such that $\int \e \phi d \lambda_t^N =1$. Then using the results on the hydrodynamic limit, along a sub-sequence, we have, for almost all $t$,
 \begin{equation}
   \begin{split}
   \liminf_{N\to\infty} \frac{H_N(t)}{N} \ \ge \ \sup_{\tilde\tau(t,\cdot), \tilde p(t,\cdot) \in \operatorname{Lip}([0,1])} \mathbb E^{\frak Q} \int_0^1 dy \Big\{
   \left[ \beta \left(\tilde\tau(t,y) -\bar\tau(t)\right) r(t,y) + \tilde p(t,y) p(t,y)\right]\\
   - G(\tilde \tau(t,y)) + G(\bar\tau(t)) - \frac{\beta}{2} \tilde p(t,y)^2 \Big\}\\
   = \mathbb E^{\frak Q} \int_0^1 dy \Big\{ \sup_{\tau \in \mathbb R} \left[\beta \tau r(t,y)  - G( \tau)\right] - \beta\bar\tau(t) r(t,y)
   + \sup_{p' \in \mathbb R} \left[p' p(t,y) - \frac{\beta}{2} p'^2\right]\Big\} + G(\bar\tau(t))\\
 = \beta\mathbb E^{\frak Q} \int_0^1 dy \left[\mathcal F(t,y) - \beta \bar\tau(t) r(t,y) \right] + G( \bar\tau(t)).
\end{split}
    \label{eq:entropy-macro1a}
  \end{equation}
The conclusion then follows after a time integration.
\end{proof}
\begin{theorem}[Clausius inequality]
\begin{equation}
  \label{eq:17}
  \begin{split}
   \mathbb E^{\mathfrak Q}\left(\int_0^T  \int_0^1 [\mathcal F(t,y)  -  \mathcal F(0,y)] dy  \right)
    \le  \mathbb E^{\mathfrak Q}\left( \int_0^TW(t)dt \right).
  \end{split}
\end{equation}
where the macroscopic work is given by 
\begin{equation}
  \label{eq:18}
 W(t) := - \int_0^t \bar \tau'(s) \int_0^1 r(s,y)dy+ \bar \tau(t) \int_0^1 r(t,y)dy- \bar \tau(0) \int_0^1 r_0(y)dy.
\end{equation}
\end{theorem}
\begin{proof}
By our assumptions on the initial conditions we have
  \begin{equation}
    \lim_{N\to\infty} \frac{H_N(0)}{N} \ =\ \beta  \int_0^1 \left[ \mathcal F(0,y) - \bar\tau(0) r_0(y)\right] dy
    + G(\bar\tau(0)).
    \label{eq:entropy-macro0}
  \end{equation}
This, together with the previous lemma, yields for almost all $t$,
  \begin{equation}
    \label{eq:15}
    \begin{split}
    \beta \mathbb E^{\mathfrak Q}\left(  \int_0^1 \mathcal F(t,y) dy  \right) - \beta  \int_0^1 \mathcal F(0,y) dy
    \le \liminf_{N\to\infty} \frac{H_N(t) - H_N(0)}{N} \\
    + \beta \mathbb E^{\mathfrak Q}\left(\int \left[\bar\tau(t) r(t,y) - \bar\tau(0) r_0(y) \right] dy\right)
    - G(\bar\tau(t)) + G(\bar\tau(0)).
  \end{split}
\end{equation}
Finally, from Proposition \ref{prop:relative} we have
\begin{align}
  \label{eq:16}
    \liminf_{N\to\infty}& \frac{H_N(t) - H_N(0)}{N} \le -\beta \int_0^t ds\;
    \mathbb E^{\mathfrak Q}\left( \bar\tau'(s) \int_0^1 r(s,y)  dy\right) 
    +\beta \int_0^t ds\;  \bar\tau'(s)\ell(\bar\tau(s)) \\
   & = \beta \mathbb E^{\mathfrak Q}\left( W(t) \right) \nonumber
    - \beta \mathbb E^{\mathfrak Q}\left(\int_0^1 \left[\bar\tau(t) r(t,y) - \bar\tau(0) r_0(y) \right] dy\right)
    +\beta \int_0^t ds\;  \bar\tau'(s)\ell(\bar\tau(s)) \nonumber\\
    &= \beta \mathbb E^{\mathfrak Q}\left( W(t) \right)
    - \beta \mathbb E^{\mathfrak Q}\left(\int_0^1 \left[\bar\tau(t) r(t,y) - \bar\tau(0) r_0(y) \right] dy\right)
    +G(\bar\tau(t)) - G(\bar\tau(0)),\nonumber
\end{align}
which, together with \eqref{eq:15} and after an integration in time gives the conclusion.
\end{proof}
\begin{oss}
Assume that the external tension varies smoothly from $\tau_0$ at $t=0$ to $\tau_1$ as $t \to \infty$. Assume also that the system is at equilibrium both at time zero and as $t \to \infty$:
\begin{align}
p_0(x) = 0, \quad \tau(r_0(x)) = \tau_0, \qquad \forall x \in [0,1]
\end{align}
\begin{align}
\lim_{t \to \infty} p(t,x) =0, \quad \lim_{t \to \infty} \tau (r(t,x)) = \tau_1 \qquad \text{a.e. } x \in [0,1].
\end{align}
Then, the following version of the Clausius inequality holds
\begin{align}
F( \tau_1) - F(\tau_0) \le E^{\mathfrak Q}(W),
\end{align}
where the total work $W$ is given by
\begin{align}
W := - \int_0^\infty \bar \tau'(s) \int_0^1 r(s,y)dy+\tau_1  \ell(\tau_1)-  \tau_0  \ell(\tau_0).
\end{align}
\end{oss}

\appendix
\section{Appendix}

\subsection{Microscopic estimates} \label{sec:estimates}
In the following we will denote, for any sequence $(a_i)_{i \in \mathbb N}$
and any $l \in \mathbb N$ the usual block averages
\begin{align}
\bar a_{l,i}:= \frac{1}{l} \sum_{j=1}^l a_{i-j+1},  \quad  i \ge l.
\end{align}
For $1\le m \le i \le N$, denote by
$d\mu_{m,i,t} \in \mathcal M\left(\mathbb R^{2m} \right)$
the projection of the probability measure $\mu_t^N$ on $\{r_{i-m+1}, \dots, r_i, p_{i-m+1}, \dots, p_i\}$.
This decompose in $d\mu_{m,i,t}(\cdot) = d\mu_{m,i,t}(\cdot|\bar r_{m,i}, \bar p_{m,i}) d\mu_{i,t}(\bar r_{m,i}, \bar p_{m,i})$,
where $\mu_{m,i,t}(\cdot|\bar r_{m,i}, \bar p_{m,i})$ is the measure $\mu_{m,i,t}$ 
conditioned to  $\bar r_{m,i}, \bar p_{m,i}$, while $d\mu_{i,t}(\bar r_{m,i}, \bar p_{m,i})$ is the marginal distribution
of $(\bar r_{m,i}, \bar p_{m,i})$ under $\mu^N_t$. 

Correspondingly, from the measure $\lambda_t^N$, we define
$d\lambda_{m,i,t}(\cdot) = d\lambda_{m,i,t}(\cdot|\bar r_{m,i}, \bar p_{m,i}) d\lambda_{i,t}(\bar r_{m,i}, \bar p_{m,i})$.

For every value of $\ell, \bar p$, we can choose the corresponding regular conditional probabilities
$\bar \mu^{\ell, \bar p}_{m,i,t} = \mu_{m,i,t}(\cdot|\bar r_{m,i}= \ell, \bar p_{m,i}= \bar p)$
and $\bar \lambda^{\ell, \bar p}_{m,i,t} =  \lambda_{m,i,t}(\cdot|\bar r_{m,i}= \ell, \bar p_{m,i}= \bar p)$.
These are probability measures
on $\mathbb R^{2m}$ supported on the $2(m-1)$ hyperplane $\Sigma_m(\ell,\bar p)$ defined by
$\frac{1}{m} \sum_{j=0}^{m-1} r_{j} = \ell, \frac{1}{m} \sum_{j=0}^{m-1} p_{j} = \bar p$.
Observe that $\bar \lambda^{\ell, \bar p}_{m,i,t}$ does not depend on $\taubar(t)$ nor on $i$.


Since the potential $V$ is uniformly convex
, the Bakry-Emery criterion applies and we have the following logarithmic Sobolev inequality (LSI)
\begin{align} \label{eq:LSI}
  \int_{\Sigma_m(\ell,\bar p)} g^2 \log g^2 d \bar \lambda^{\ell, \bar p}_{m} \le
  C_{lsi} m^2 \sum_{j=1}^{m-1} \int_{\Sigma_m(\ell,\bar p)}
  \left[ (D_{j} g)^2+ (\tilde D_{j}g)^2 \right] d \bar \lambda^{\ell, \bar p}_{m}
\end{align}
for any smooth $g$ on $\mathbb R^{2m}$ such that $\int_{\Sigma_m(\ell,\bar p)} g^2 d \bar \lambda^{\ell, \bar p}_{m}=1$.
Here $C_{lsi}$ is a universal constant depending on the interaction $V$ only.
In particular \eqref{eq:LSI} holds for
$g^2 = \bar f_{m,i,t}^{\ell, \bar p} = d \bar \mu_{m,i,t}^{\ell, \bar p}/d \lambda_{m}^{\ell, \bar p}$.
Denote $d\mu_{m,i,t}(\ell, \bar p)$ the marginal distribution of  $\bar r_{m,i}, \bar p_{m,i}$ of $\mu_t$.

\begin{lem} \thlabel{lem:LSI}
Let $m < i < N$. Then there exists a constant $C$ such that, for any $\ell, \bar p \in \mathbb R$,
\begin{align} \label{eq:LSIlem}
  \int_{\mathbb R^2} d\mu_{m,i,t}(\bar r_{m,i}, \bar p_{m,i})
  \int_{\Sigma_m(\bar r_{m,i}, \bar p_{m,i})} \bar f_{m,i,t}^{\bar r_{m,i}, \bar p_{m,i}} \log \bar f_{m,i,t}^{\bar r_{m,i}, \bar p_{m,i}}
  d \lambda_{m}^{\bar r_{m,i}, \bar p_{m,i}} \\
  \le C_{lsi} m^2 \sum_{j=1}^{m-1}\int_{\mathbb R^{2N}}
  \frac{ \left(D_{i-j}f_t^N\right)^2+\left(\tilde D_{i-j}f_t^N\right)^2}{ f_t^N} d\lambda_t^N  \nonumber
\end{align}
\end{lem}
\begin{proof}
Let $f_{m,i,t} = \dfrac{d\mu_{m,i,t}}{d\lambda_{m,i,t}}$. It decomposes as
\begin{equation}
  \label{eq:7}
  f_{m,i,t}(r_{i-m}, p_{i-m}, \dots, r_{i}, p_{i}) = \bar f_{m,i,t}^{\bar r_{m,i}, \bar p_{m,i}}(r_{i-m}, p_{i-m}, \dots, r_{i}, p_{i})
  \bar f_{i,t}(\bar r_{m,i}, \bar p_{m,i}),
\end{equation}
where $\bar f_{i,t}=\dfrac{d\mu_{i,t}}{d\lambda_{i,t}}$. 

By \eqref{eq:LSI}, the left hand side of \eqref{eq:LSIlem} is less or equal to
\begin{align}
 C_{lsi} m^2 &\int_{\mathbb R^2} d\mu_{i,t}(\bar r_{m,i}, \bar p_{m,i})
  \sum_{j=1}^{m-1}  \int
  \frac{ (D_{i-j} \bar f_{m,i,t}^{\bar r_{m,i}, \bar p_{m,i}})^2+ (\tilde D_{i-j}\bar f_{m,i,t}^{\bar r_{m,i}, \bar p_{m,i}})^2}
  {\bar f_{m,i,t}^{\bar r_{m,i}, \bar p_{m,i}} } d \bar \lambda_{m,i,t}^{\bar r_{i,m}, \bar p_{i,m}}\\
  = C_{lsi} m^2& \int_{\mathbb R^2} f_{i,t}(\bar r_{m,i}, \bar p_{m,i}) d\lambda_{i,t}(\bar r_{m,i}, \bar p_{m,i})
  \sum_{j=1}^{m-1}  \int
  \frac{ (D_{i-j} \bar f_{m,i,t}^{\bar r_{m,i}, \bar p_{m,i}})^2+ (\tilde D_{i-j}\bar f_{m,i,t}^{\bar r_{m,i}, \bar p_{m,i}})^2}
  {\bar f_{m,i,t}^{\bar r_{m,i}, \bar p_{m,i}} } d \bar \lambda_{m,i,t}^{\bar r_{i,m}, \bar p_{i,m}} \nonumber 
  \\
  = C_{lsi} m^2 &\int_{\mathbb R^2} d\lambda_{i,t}(\bar r_{m,i}, \bar p_{m,i})
  \sum_{j=1}^{m-1}  \int
  \frac{ (D_{i-j} f_{m,i,t})^2+ (\tilde D_{i-j} f_{m,i,t})^2}
  { f_{m,i,t} } d \bar \lambda_{m,i,t}^{\bar r_{i,m}, \bar p_{i,m}}\nonumber
    \\
  = C_{lsi} m^2 
  &\sum_{j=1}^{m-1}  \int
  \frac{ (D_{i-j} f_{m,i,t})^2+ (\tilde D_{i-j} f_{m,i,t})^2}
  { f_{m,i,t} } d \lambda_{m,i,t}, \nonumber
\end{align}
and \eqref{eq:LSIlem} follows by Jensen's inequality.
\end{proof}

\begin{prop}[One-block estimate - interior] \thlabel{prop:1block}
There exists $l_0 \in \mathbb N$ such that, for $l_0 < l \le N$, we have
\begin{equation} \label{eq:oneblock}
\sum_{i=l+1}^{N-1} \int_0^t \int \left( \bar V'_{l,i} - \tau \left(\bar r_{l,i} \right) \right)^2 d \mu_s^N ds \le C \left( \frac{N}{l} t+ \frac{l^2}{\sigma} \right)
\le C(T) \frac{l^2}{\sigma}.
\end{equation}
\end{prop}
\begin{proof}
Fix $\alpha >0$. By the entropy inequality and  \thref{lem:LSI}:
\begin{align} \label{eq:1block1}
&\sum_{i=l+1}^{N-1}\alpha l \int_0^t   ds \int \left( \bar V'_{l,i} - \tau \left(\bar r_{l,i} \right) \right)^2 d \mu_s^N
\\
  =&\sum_{i=l+1}^{N-1} \alpha l \int_0^t ds \int  d\mu_{i,s}(\bar r_{l,i}, \bar p_{l,i})
    \int \left( \bar V'_{l,i} - \tau \left(\bar r_{l,i} \right) \right)^2 d\bar \mu_{l,i,s}^{\bar r_{l,i}, \bar p_{l,i}} \nonumber
 \\
\le& l^3 C_{lsi}\int_0^t (\mathcal D_N(s)+\widetilde {\mathcal D}_N(s)) d s 
        +\sum_{i=l+1}^{N-1} \int_0^t ds \int d\mu_{i,s}(\bar r_{l,i}, \bar p_{l,i})
        \log\left(  \int \e{\alpha l \left( \bar V'_{l,i}- \tau(\bar r_{l,i})\right)^2} d \bar \lambda_{l,i,t}^{\bar r_{l,i}, \bar p_{l,i}} \right) \nonumber
 \\
 \le& C\frac{l^3}{\sigma}+  \sum_{i=l+1}^{N-1} \int_0^t ds \int d\mu_{i,s}(\bar r_{l,i}, \bar p_{l,i})
        \log\left(  \int \e{\alpha l \left( \bar V'_{l,i}- \tau(\bar r_{l,i})\right)^2} d \bar \lambda_{l}^{\bar r_{l,i}, \bar p_{l,i}} \right) \nonumber
\end{align}
where we have used the bound on the time integral of the Dirichlet form 
.

We prove now that for $\alpha < (4c_0)^{-1}$ we have
\begin{equation}
  \label{eq:8}
  \sup_{\ell,\bar p \in \R} \int \e{\alpha l \left( \bar V'_{l,i}- \tau(\ell)\right)^2} d \bar \lambda_{l}^{\ell, \bar p} \le C,
\end{equation}
and \eqref{eq:oneblock} will follow.

We take $l > l_0$ so that
\begin{equation}
  \label{eq:9}
  \int \e{\alpha l \left( \bar V'_{l,i}- \tau(\ell)\right)^2} d \bar \lambda_{l}^{\ell, \bar p} \le C
  \int \e{\alpha l \left( \bar V'_{l,i}- \tau(\ell)\right)^2} d \lambda^l_{\beta, \bar p, \tau(\ell)}.
\end{equation}


Let us introduce a normally distributed random variable $\xi \sim \mathcal N(0,1)$ so that we can use the identity
\begin{align}
  \e{\alpha l \left( \bar V'_{l,i}- \tau(\ell)\right)^2}
  = \mathbb E_\xi \left[ \e {\xi \sqrt{2 \alpha l}\left(\bar V'_{l,i}-\tau(\ell) \right) } \right]
\end{align}
in order to write
\begin{align}
  \int \e{\alpha l \left( \bar V'_{l,i}- \tau(\ell)\right)^2} d \lambda^l_{\beta, \bar p,\tau(\ell)} &
 = \mathbb E_\xi \left[ \int \e {\xi \sqrt{2 \alpha l}\left(\bar V'_{l,i}-\tau(\ell) \right) }d \lambda^l_{\beta, \bar p,\tau(\ell)} \right] 
  \\ & = \mathbb E_\xi \left[\e{-\tau(\ell) \xi \sqrt{2\alpha l}} \left(\int \e{  \frac{\xi \sqrt{2\alpha l}}{l}V'(r) }
       d \lambda_{\beta, \bar p,\tau(\ell)} \right)^l \right],
\end{align}
It is easy to show (cf Appendix A of \cite{marchesani2018hydrodynamic}) that,
\begin{equation}
  \int \e {\xi \sqrt{2\alpha l^{-1}} V'(r_1) } d \lambda^N_{\beta, \bar p,\tau(\ell)} \le
  \e  {\frac{c_2\alpha}{\beta}\xi^2+ \frac{\tau(\ell) \sqrt{2\alpha l}}{l}\xi}.
\end{equation}
Hence, we obtain
\begin{equation}
  \int \e{\alpha \left( \bar V'_{l,i}- \tau(\ell)\right)^2} d  \lambda^l_{\beta, \bar p,\tau(\ell)}
  \le  \mathbb E_\xi \left[ \e{ \frac{c_2 \alpha}{\beta} \xi^2} \right],
\end{equation}
and the right hand side is independent of $\ell$ and $\bar p$.
Taking $\alpha = \beta/(4c_2)$ the expectation with respect to $\xi$ is finite.
\end{proof}
If we do not perform the summation over $i$ in \eqref{eq:1block1} and we bound the right-hand side of \eqref{eq:LSIlem} by the Dirichlet forms we obtain the following
\begin{cor} \thlabel{prop:11block}
There exists $l_0 \in \mathbb N$ such that, for $l_0 < l < i < N$, we have
\begin{equation} \label{eq:oneblock}
 \int_0^t \int \left( \bar V'_{l,i} - \tau \left(\bar r_{l,i} \right) \right)^2 d \mu_s^N ds \le C \left( \frac{1}{l} t+ \frac{l}{\sigma} \right)
\le C(T) \frac{l}{\sigma}.
\end{equation}
\end{cor}

\begin{prop}[One-and-a-half-block estimate] \thlabel{prop:1.5block}
Let $ a_i \in \{ p_i,  V'(r_i)\}$. Then, for any fixed $ l \le k \le N-l$, we have
\begin{align}
\int_0^t\int (\bar a_{l,k+l}-\bar a_{l,k})^2 d \mu_s^Nds \le
 C \left(\frac{1}{l}t + \frac{l}{\sigma} \right)
 \le C(T) \frac{l}{\sigma}.
\end{align}
\end{prop}
\begin{proof}
We consider $a_i = V'(r_i)$, as the case $a_i=p_i$ is analogous. Thanks to the identity
\begin{align}
V'(r_i)-V'(r_j) = - \beta^{-1} \left( \frac{\partial \lambda_s^N}{\partial r_i}-\frac{\partial \lambda_s^N}{\partial r_j}\right)
\end{align}
we compute
\begin{align}
\int (\bar V'_{l,k+l}-\bar V'_{l,k})^2 f_s^N d \lambda_s^N = & \frac{1}{l} \sum_{i=k-l+1}^k\int(V'(r_{i+l})-V'(r_i))(\bar V'_{l,k+l}-\bar V'_{l,k}) f_s^N d \lambda_s^N
\\
 =& \frac{\beta^{-1}}{l^2} \sum_{i=k-l+1}^k \int (V''(r_{i+l})+V''(r_i)) d \mu_s^N + \nonumber
\\
&+ \int (\bar V'_{l,k+l}-\bar V'_{l,k}) \left( \frac{\beta^{-1}}{l} \sum_{i=k-l+1}^k \left( \frac{\partial f_s^N}{\partial r_{i+l}}-\frac{\partial f_s^N}{\partial r_i}\right) \right) d \lambda_s^N. \nonumber
\end{align}
By using Cauchy-Schwartz inequality on the last term and the fact that $V''$ is bounded, we obtain
\begin{align}
\int (\bar V'_{l,k+l}-\bar V'_{l,k})^2 f_s^N d \lambda_s^N & \le C \left(\frac{1}{l} + \frac{1}{l} \int \frac{1}{f_s^N}\sum_{i=k-l+1}^k \left( \frac{\partial f_s^N}{\partial r_{i+l}}-\frac{\partial f_s^N}{\partial r_i} \right)^2  d\lambda_s^N\right) \nonumber
\\
& = C \left( \frac{1}{l} + \frac{1}{l}\int \frac{1}{f_s^N} \sum_{i=k-l+1}^k \left( \sum_{j=i}^{i+l-1} \left( \frac{\partial f_s^N}{\partial r_{j+1}}-\frac{\partial f_s^N}{\partial r_j}\right) \right)^2d \lambda_s^N \right) \nonumber
\\
& \le C \left( \frac{1}{l} + \int \frac{1}{f_s^N}\sum_{i=k-l+1}^k \sum_{j=i}^{i+l-1}\left( \frac{\partial f_s^N}{\partial r_{j+1}}-\frac{\partial f_s^N}{\partial r_j}\right)^2d \lambda_s^N \right) \label{eq:1.5block}
\\
& \le C \left( \frac{1}{l}+\int \frac{1}{f_s^N} \sum_{i=k-l+1}^k \sum_{j=1}^{N-1} \left( \frac{\partial f_s^N}{\partial r_{j+1}}-\frac{\partial f_s^N}{\partial r_j}\right)^2 d \lambda_s^N \right) \nonumber
\\
& \le C \left(\frac{1}{l}+ l \tilde{\mathcal D}_N(s) \right). \nonumber
\end{align}
The conclusion then follows after an integration in time.
\end{proof}

\begin{prop}[Two-block estimate] \thlabel{prop:2block}
Let $ a_i \in \{ p_i, V'(r_i), r_i \}$ and $l_0$ be as in \thref{prop:1block}. Then, for $l_0 < l \le  N$, we have
\begin{align}
\sum_{i=l+1}^{N-l} \int_0^t \left( \bar a_{l,i+l}- \bar a_{l,i} \right)^2 d \mu_s^N ds \le C \left( \frac{N}{l} t+ \frac{l^2}{\sigma} \right)
\le C(T) \frac{l^2}{\sigma}.
\end{align}
\end{prop}
\begin{proof}
We prove the statement for $ a_i = V'(r_i)$. From \eqref{eq:1.5block} we have
\begin{align}
\sum_{i=l+1}^{N-l} \int (\bar V'_{l,i+l}-\bar V'_{l,i})^2 f_s d \lambda_s^N & \le C \left( \frac{N}{l} +\int \frac{1}{f_s^N} \sum_{i=l}^{N-l} \sum_{k=i-l+1}^i \sum_{j=k}^{k+l-1} \left( \frac{\partial f_s^N}{\partial r_{j+1}} -\frac{\partial f_s^N}{\partial r_j} \right)^2 d \lambda_s^N\right) \nonumber
\\
& \le C \left( \frac{N}{l} +l \int \frac{1}{f_s^N} \sum_{i=l}^{N-l} \sum_{j=i-l+1}^{i+l-1}\left( \frac{\partial f_s^N}{\partial r_{j+1}} -\frac{\partial f_s^N}{\partial r_j} \right)^2 d \lambda_s^N \right) \nonumber
\\
& \le C \left( \frac{N}{l} +l^2 \int \frac{1}{f_s^N} \sum_{i=l}^{N-l}\left( \frac{\partial f_s^N}{\partial r_{i+1}} -\frac{\partial f_s^N}{\partial r_i} \right)^2 d \lambda_s^N \right)
\\
& \le C \left( \frac{N}{l} + l^2 \tilde{\mathcal D}_N(s) \right). \nonumber
\end{align}
The conclusion then follows after a time integration. The proof for $a_i= p_i$ is analogous. Finally, since $\tau'$ is bounded from below by a positive constant, we have
\begin{align}
(\bar r_{l,i+l}-\bar r_{l,i})^2 &\le C \left( \tau(\bar r_{l,i+l})-\tau(\bar r_{l,i})\right)^2
\\
& \le C \left[ \left( \tau(\bar r_{l,i+l})-\bar V'_{l,i+l}\right)^2 + \left( \bar V'_{l,i}-\tau(\bar r_{l,i})\right)^2 + \left(\bar V'_{l,i+l}-\bar V'_{l,i}\right)^2 \right] \nonumber
\end{align}
and the statement for $a_i = r_i$ follows from the fist part of the proof and \thref{prop:1block}.
\end{proof}

We conclude this section by showing the connection between the averages $\hat a_{l,i}$ and $\bar a_{l,i}$.
\begin{lem} \thlabel{lem:2blockhat}
For any sequence $(a_i)_{i \in \mathbb N}$, any $l \in \mathbb N$ and any $i \ge l$, we have
\begin{align}
\nabla \hat a_{l,i}= \frac{1}{l}(\bar a_{l,i+l}-\bar a_{l,i})
\end{align}
\end{lem}
\begin{proof}
We prove the statement by induction over $l$. The statement for $l=1$ is obvious, since both $\hat a_{1,i+1}-\hat a_{1,i}$ and $\bar a_{1,i+1}-\bar a_{1,i}$ are equal to $a_{i+1} -a_i$.

Assume now the statement is true for some $l \ge 1$. We prove it holds for $l+1$ as well. We have
\begin{align}
\hat a_{l+1,i+1}-\hat a_{l+1,i}& =  \frac{1}{l+1} \sum_{|j| <l+1} \frac{l+1-|j|}{l+1}a_{i+1-j}-\frac{1}{l+1} \sum_{|j| <l+1} \frac{l+1-|j|}{l+1}a_{i-j}
\\
&= \frac{1}{(l+1)^2} \sum_{|j| <l+1} (l+1-|j|)(a_{i+1-j}-a_{i-j}) \nonumber
\\
&= \frac{l^2}{(l+1)^2} \frac{1}{l} \sum_{|j|<l} \frac{l-|j|}{l}(a_{i+1-j}-a_{i-j}) +\frac{1}{(l+1)^2} \sum_{|j| < l+1} (a_{i+1-j}-a_{i-j}). \nonumber
\end{align}
For the first summation we can use the inductive hypothesis, while the second summation is telescopic. Therefore we obtain
\begin{align}
\hat a_{l+1,i+1}-\hat a_{l+1,i} &= \frac{1}{(l+1)^2} \sum_{j=1}^l (a_{i+l-j+1}-a_{i-j+1}) + \frac{1}{(l+1)^2} (a_{i+l+1}- a_{i-l})
\\
&= \frac{1}{(l+1)^2} \sum_{j=1}^{l+1} (a_{i+l-j+1}-a_{i-j+1}) \nonumber
\\
&= \frac{1}{l+1}(\bar a_{l+1,i+l+1}-\bar a_{l+1,i}). \nonumber
\end{align}
\end{proof}
Combining \thref{prop:2block} and \thref{lem:2blockhat} we get the following
\begin{cor} \thlabel{cor:2block}
Let  $a_i \in \{ p_i, V'(r_i), r_i\}$ and $l_0$ be as in \thref{prop:1block}. Then, for $l_0 < l<i   <N-l+1$,
\begin{align}
\sum_{j=l+1}^{N-l} \int_0^t\int \left(\nabla \hat a_{l,j} \right)^2 d \mu_s^N ds &\le C \left( \frac{N}{l^3} t+ \frac{1}{\sigma} \right)
\le C(T) \frac{1}{\sigma}
\\
 \int_0^t\int \left(\nabla \hat a_{l,i} \right)^2 d \mu_s^N ds& \le C \left( \frac{1}{l^3} t+ \frac{1}{l\sigma} \right)
\le C(T) \frac{1}{l\sigma}.
\end{align}

We now show that the two averages we defined are equivalent in the limit.
\begin{prop}[One-block comparison]
Let $a_i \in \{p_i, V'(r_i)\}$. Then, for any $l \le i \le N-l$, we have
\begin{align}
\int_0^t\int  \left( \bar a_{l,i}- \hat a_{l,i} \right)^2 d \mu_s^N ds \le C \left( \frac{1}{l} t+ \frac{l}{\sigma} \right)
\le C(T) \frac{l}{\sigma}.
\end{align}
\end{prop}
\begin{proof}
We prove the statement for $a_i = V'(r_i)$, the proof for $a_i = p_i$ being analogous. We can write
\begin{align}
\bar a_{l,i}-\hat a_{l,i} = \frac{1}{l} \sum_{j=0}^{l-1} \frac{j}{l} \left( a_{i-j}-a_{i-j+l} \right).
\end{align}
Thus,
\begin{align}
&\int (\bar V'_{l,i}-\hat V'_{l,i})^2 d \mu_s^N  = \int\frac{1}{l} \sum_{j=0}^{l-1} \frac{j}{l} \left(V'(r_{i-j})-V'(r_{i-j+l})\right) (\bar V'_{l,i}-\hat V'_{l,i}) f_s^N d \lambda_s^N
\\
 =& \beta^{-1} \frac{1}{l^2} \int \sum_{j=0}^{l-1} \sum_{k=0}^{l-1} \frac{jk}{l^2}\left( \left( \frac{\partial}{\partial r_{i-j}}-\frac{\partial}{\partial r_{i-j+l}} \right)  \left(V'(r_{i-k})-V'(r_{i-k+l}) \right)\right)f_s^N d\lambda_s^N+ \nonumber
\\
& \nonumber + \beta^{-1} \int (\bar V'_{l,i}-\hat V'_{l,i}) \sqrt{f_s^N} \frac{1}{l} \sum_{j=0}^{l-1} \frac{j}{l} \frac{1}{\sqrt{f_s^N}} \left(  \frac{\partial f_s^N}{\partial r_{i-j}}-\frac{\partial f_s^N}{\partial r_{i-j+l}} \right) d \lambda_s^N \nonumber
\\
 \le& \beta^{-1} \frac{1}{l^2} \int \sum_{j=0}^{l-1} \frac{j^2}{l^2} \left( V''(r_{i-j})+V''(r_{i-j+l})\right) f_s^N d \lambda_s^N+ \nonumber
\\
&+ \frac{1}{2} \int (\bar V'_{l,i}-\hat V'_{l,i})^2 f_s^N d \lambda_s^N + \frac{\beta^{-2}}{2} \int \frac{1}{l} \sum_{j=0}^{l-1} \frac{j^2}{l^2} \frac{1}{f_s^N} \left( \frac{\partial f_s^N}{\partial r_{i-j}}-\frac{\partial f_s^N}{\partial r_{i-j+l}} \right)^2 d \lambda_s^N. \nonumber
\end{align}
Thus, we obtain
\begin{align}
\int (\bar V'_{l,i}-\hat V'_{l,i})^2 d \mu_s^N \le 2\beta^{-1} \frac{\|V''\|_\infty}{l} + \beta^{-2} \int \frac{1}{l} \sum_{j=0}^{l-1} \frac{1}{f_s^N} \left( \frac{\partial f_s^N}{\partial r_{i-j}}-\frac{\partial f_s^N}{\partial r_{i-j+l}} \right)^2 d \lambda_s^N
\end{align}
and the conclusion follows as in the proof of \thref{prop:1.5block}.
\end{proof}


\begin{prop}[Block average comparison]\thlabel{prop:averages}
Let  $a_i \in \{p_i, V'(r_i), r_i\}$ and $l_0$ as in \thref{prop:1block}. Then, for $l_0 < l \le N$ we have
\begin{align}
\sum_{i=l+1}^{N-l} \int_0^t \int \left( \bar a_{l,i}- \hat a_{l,i} \right)^2 d \mu_s^N ds \le C \left( \frac{N}{l} t+ \frac{l^2}{\sigma} \right)
\le C(T) \frac{l^2}{\sigma}.
\end{align}
\end{prop}
\begin{proof}
The statement for $a_i \in \{p_i, V'(r_i)\}$ is obtained by the previous proposition after summing over $i$ as in the proof of \thref{prop:2block}.

In order to prove the statement for $a_i= r_i$ we write
\begin{align}
\bar r_{l,i}-\hat r_{l,i} = \frac{1}{l} \sum_{|j|<l} c_j r_{i-j},
\end{align}
where $|c_j|<1$ and $\sum_{|j|<l}c_j =0$. Let $\alpha >0$ to be chosen later and follow the proof of \thref{prop:1block} in order to obtain
\begin{align} \label{eq:rav1}
\alpha \sum_{i=l+1}^{N-l} \int_0^t \int \left( \bar r_{l,i}- \hat r_{l,i} \right)^2 d \mu_s^N ds &\le C\frac{l^3}{\sigma}+  \sum_{i=l+1}^{N-l} \int_0^t\int \log\left(  \int\e{\alpha \left(\bar r_{l,i}-\hat r_{l,i}\right)^2} d\hat \lambda_{2l-1,i+l-1}^{\ell, \bar p}\right)d \mu_s^Nds.
\end{align}
We introduce a normally distributed $\xi \sim \mathcal N(0,1)$ and write, for large enough $l$,
\begin{align}
 \int\e{\alpha \left(\bar r_{l,i}-\hat r_{l,i}\right)^2} d\hat \lambda_{2l-1,i+l-1}^{\ell, \bar p} & \le C  \int\e{\alpha \left(\bar r_{l,i}-\hat r_{l,i}\right)^2} d\hat \lambda^N_{\beta,\bar p,\tau(\ell)}
\\ 
& =C \mathbb E_\xi \left[  \int\e{\sqrt{2\alpha} \xi \left(\bar r_{l,i}-\hat r_{l,i}\right)} d\hat \lambda^N_{\beta,\bar p,\tau(\ell)}\right] \nonumber
\\
& = C \mathbb E_\xi \left[  \int\e{\sqrt{2\alpha} \xi\frac{1}{l} \sum_{|j|<l} c_j r_{i-j} }d\hat \lambda^N_{\beta,\bar p,\tau(\ell)}\right] \nonumber
\\
& = C \mathbb E_\xi \left[\prod_{|j|<l}  \e{G\left( \tau(\ell)+\frac{\sqrt{2\alpha}\xi}{l}c_j\right)-G( \tau(\ell)) } \right] \nonumber
\\
& = C \mathbb E_\xi \left[ \e{ \sum_{|j|<l }  \left(  G'(\tau(\ell)) \frac{\sqrt{2\alpha}\xi}{l}c_j +  G''(\tilde \tau) \frac{2 \alpha \xi^2}{l^2}c_j^2 \right)}\right] \nonumber
\end{align}
for some intermediate value $\tilde \tau$. But since $\sum_{|j|<l} c_j=0$, $|c_j|<1$ and $ G''$ is bounded, we can estimate
\begin{align} \label{eq:rav2}
 \int\e{\alpha \left(\bar r_{l,i}-\hat r_{l,i}\right)^2} d\hat \lambda_{2l-1,i+l-1}^{\ell, \bar p} & \le C \mathbb E_\xi \left[ \e{ \frac{6 \alpha \|G''\|_\infty }{l}\xi^2} \right]  = C \frac{1}{\sqrt{1-\dfrac{12 \alpha \| G''\|_\infty}{l}}},
\end{align}
provided $\dfrac{6 \alpha \|G''\|_\infty}{l} < \dfrac{1}{2}$. Note that the right-hand side of \eqref{eq:rav2} does not depend on $\ell$ and $\bar p$. Thus, combining \eqref{eq:rav1} and \eqref{eq:rav2} and choosing $\alpha < \dfrac{1}{12\| G''\|_\infty}l$ leads to the conclusion.
\end{proof}
\begin{cor} \thlabel{prop:11block}
There exists $l_0 \in \mathbb N$ such that, for $l_0 < l < i < N-l+1$, we have
\begin{equation}
 \int_0^t \int \left( \bar r_{l,i} - \hat r_{l,i}  \right)^2 d \mu_s^N ds \le C \left( \frac{1}{l} t+ \frac{l}{\sigma} \right)
\le C(T) \frac{l}{\sigma}.
\end{equation}
\end{cor}

\begin{cor} \thlabel{cor:1block}
Let $l_0$ be as in \thref{prop:1block}. Then, for $l_0 < l < i <  N-l+1$ we have
\begin{align}
\sum_{j=l+1}^{N-l}\int_0^t \int \left(\hat V'_{l,j} - \tau(\hat r_{l,j}) \right)^2 d \mu_s^Nds &\le C \left( \frac{N}{l}t + \frac{l^2}{\sigma} \right)
\le C(T) \frac{l^2}{\sigma}
\\
\int_0^t \int \left(\hat V'_{l,i} - \tau(\hat r_{l,i}) \right)^2 d \mu_s^Nds &\le C \left( \frac{1}{l}t + \frac{l}{\sigma} \right) \nonumber
\le C(T) \frac{l}{\sigma}
\end{align}
\end{cor}
\begin{proof}
Follows from  \thref{prop:1block}, \thref{prop:averages} and the inequality
\begin{align}
\left(\hat V'_{l,i} - \tau(\hat r_{l,i}) \right)^2 & \le 3 \left[\left(\hat V'_{l,i}- \bar V'_{l,i} \right)^2 + \left( \bar V'_{l,i}-\tau(\bar r_{l,i})\right)^2 + \left( \tau (\bar r_{l,i})-\tau(\hat r_{l,i})\right)^2 \right]
\\
&  \le C \left[\left(\hat V'_{l,i}- \bar V'_{l,i} \right)^2 + \left( \bar V'_{l,i}-\tau(\bar r_{l,i})\right)^2 + \left(  \bar r_{l,i}- \hat r_{l,i}\right)^2 \right] \nonumber
\end{align}
\end{proof}

\end{cor}

\section*{Acknowledgments} 
This work has been partially supported by the grants ANR-15-CE40-0020-01 LSD 
of the French National Research Agency.

\addcontentsline{toc}{chapter}{References}

\renewcommand{\bibname}{References}
\nocite{*}
\bibliography{bibliografia_new}
\bibliographystyle{plain}

\noindent
{Stefano Olla\\
CEREMADE, UMR-CNRS, Universit\'e de Paris Dauphine, PSL Research University}\\
{\footnotesize Place du Mar\'echal De Lattre De Tassigny, 75016 Paris, France}\\
{\footnotesize \tt olla@ceremade.dauphine.fr}\\
\\
{Stefano Marchesani\\
GSSI, \\
{\footnotesize Viale F. Crispi 7, 67100 L'Aquila, Italy}}\\
{\footnotesize \tt stefano.marchesani@gssi.it}\\

\end{document}